%% file: main.tex
\documentclass[11pt]{article}

\usepackage{amsmath}
\usepackage{amssymb}
\usepackage{dsfont}
\usepackage{algorithmic}
\usepackage{fullpage}
\usepackage{multicol}
\usepackage{amsthm}
\usepackage{xcolor}
\usepackage{caption,subcaption}
\usepackage{xspace}

\usepackage[T1]{fontenc}

\usepackage{hyperref}
\hypersetup{
    colorlinks=true,
    linkcolor=[rgb]{0,0,0.6},
    citecolor=[rgb]{0,0,0.6}
}

\mathchardef\mhyphen="2D

\newcommand{\eps}{\epsilon}
\newcommand{\OPT}{\mathsf{OPT}}
\newcommand{\R}{\mathbb{R}}
\newcommand{\1}{\mathds{1}}
\newcommand{\abs}[1]{\left|#1\right|}
\newcommand{\E}{\mathbb{E}}
\newcommand{\ip}[2]{\langle #1, #2\rangle}
\DeclareMathOperator*{\argmin}{arg\,min}

\newcommand{\cC}{\mathcal{C}}
\newcommand{\cI}{\mathcal{I}}
\newcommand{\cM}{\mathcal{M}}
\newcommand{\cP}{\mathcal{P}}
\newcommand{\cS}{\mathcal{S}}
\newcommand{\cT}{\mathcal{T}}
\newcommand{\cU}{\mathcal{U}}
\newcommand{\poly}{\operatorname{poly}}

\newcommand{\tO}{\tilde{O}}

\newcommand{\ST}{\mathsf{ST}}
\newcommand{\setdiff}{\triangle}
\newcommand{\ceil}[1]{\left\lceil#1\right\rceil}

\newcommand{\eqdist}{\stackrel{d}{=}}

\newcommand{\bdmst}{{\sc BD-MST}\xspace}
\newcommand{\bdst}{{\sc BD-ST}\xspace}
\newcommand{\cst}{{\sc Crossing-ST}\xspace}
\newcommand{\cmst}{{\sc Crossing-MST}\xspace}

\newcommand{\crossrand}{\mathsf{crossing\mhyphen
    rand \mhyphen MWU}}
\newcommand{\fastcrossrand}{\mathsf{fast\mhyphen crossing\mhyphen
    rand \mhyphen MWU}}

\newcommand{\divideswap}{\mathsf{divide\mhyphen and \mhyphen conquer
  \mhyphen swap}}

\newcommand{\mstlength}{\rho}
\newcommand{\approxmstlength}{\bar{\mstlength}}
\newcommand{\constcost}{\gamma}

\newcommand{\Nguyen}{Nguy\~{\^{e}}n\xspace}
\newcommand{\Furer}{F{\"u}rer\xspace}

\newcommand{\reduce}{\mathsf{Reduce}}

\newcommand{\randpack}{\mathsf{rand \mhyphen mwu \mhyphen pack}}

\newcommand{\merge}{\mathsf{merge\mhyphen bases}}
\newcommand{\swap}{\mathsf{swap\mhyphen round}}
\newcommand{\fastmerge}{\mathsf{fast\mhyphen merge}}
\newcommand{\fastswap}{\mathsf{fast\mhyphen swap}}

\newcommand{\initedge}{\mathsf{init \mhyphen edges}}

\newcommand{\shrinkint}{\mathsf{shrink \mhyphen intersection}}

\newcommand{\ins}{\mathsf{insert}}

\newcommand{\makeset}{\mathsf{make \mhyphen set}}
\newcommand{\changerep}{\mathsf{change \mhyphen rep}}
\newcommand{\findset}{\mathsf{find \mhyphen set}}
\newcommand{\union}{\mathsf{union}}

\newcommand{\maketree}{\mathsf{make \mhyphen tree}}
\newcommand{\link}{\mathsf{link}}
\newcommand{\cut}{\mathsf{cut}}

\newcommand{\evert}{\mathsf{evert}}
\newcommand{\join}{\mathsf{join}}
\newcommand{\findroot}{\mathsf{find \mhyphen root}}

\newcommand{\init}{\mathsf{init}}
\newcommand{\origedge}{\mathsf{orig\mhyphen edge}}
\newcommand{\represented}{\mathsf{represented \mhyphen edge}}
\newcommand{\neighbors}{\mathsf{neighbors}}
\newcommand{\leafandneighbor}{\mathsf{leaf \mhyphen and
  \mhyphen neighbor}}
\newcommand{\remove}{\mathsf{remove}}
\newcommand{\contract}{\mathsf{contract}}
\newcommand{\leaving}{\mathsf{first\mhyphen edge}}
\newcommand{\copyy}{\mathsf{copy}}

\definecolor{dgreen}{RGB}{34,139,34}

\newtheorem{observation}{Observation}[section]
\newtheorem{theorem}[observation]{Theorem}
\newtheorem{lemma}[observation]{Lemma}

\newtheorem{remark}[observation]{Remark}

\newtheorem{corollary}[observation]{Corollary}

\newtheorem{manualtheoreminner}{Theorem}
\newenvironment{manualtheorem}[1]{%
  \renewcommand\themanualtheoreminner{#1}%
  \manualtheoreminner
}{\endmanualtheoreminner}

\newcommand{\etal}{et al.\ }

\author{
Chandra Chekuri\thanks{Dept.\ of Computer Science, Univ.\ of Illinois, Urbana-Champaign, Urbana,
  IL 61801. {\tt chekuri@illinois.edu}. Supported in part by NSF grant
CCF-1910149.}
\and
Kent Quanrud\thanks{Dept.\ of Computer Science, Purdue University,
  West Lafayette, IN 47909. {\tt krq@purdue.edu}.}
\and
Manuel R.\ Torres\thanks{Dept.\ of Computer Science, Univ.\ of Illinois, Urbana-Champaign, Urbana,
  IL 61801. {\tt manuelt2@illinois.edu}. Supported in part by fellowships from NSF
    and the Sloan Foundation, and NSF grant CCF-1910149.}
}
\title{Fast Approximation Algorithms
for Bounded Degree and Crossing Spanning Tree Problems}

\begin{document}
\hypersetup{pageanchor=false}
\maketitle

\abstract{
We develop fast approximation algorithms for the
  minimum-cost version of the Bounded-Degree MST problem (\bdmst) and
  its generalization the Crossing Spanning Tree problem (\cst). We
  solve the underlying LP to within a $(1+\eps)$ approximation factor
  in near-linear time via the multiplicative weight update (MWU)
  technique. This yields, in particular, a near-linear time algorithm
  that outputs an estimate $B$ such that $B \le B^* \le \ceil{(1+\eps)B}+1$ where
  $B^*$ is the minimum-degree of a spanning tree of a given graph. To
  round the fractional solution, in our main technical contribution,
  we describe a fast near-linear time implementation of swap-rounding
  in the spanning tree polytope of a graph. The fractional solution
  can also be used to sparsify the input graph that can in turn be
  used to speed up existing combinatorial algorithms. Together, these
  ideas lead to significantly faster approximation algorithms than
  known before for the two problems of interest. In addition,
  a fast algorithm for swap rounding in the graphic matroid is a generic
  tool that has other applications, including to TSP and submodular function
  maximization.
}

\thispagestyle{empty}
\newpage
\setcounter{page}{1}

\hypersetup{pageanchor=true}

\input{intro.tex}

\input{prelim.tex}

\input{swaprounding.tex}

\input{sparsification.tex}

\input{lp-solve.tex}

\input{putting-things-together.tex}

\paragraph{Acknowledgements:} We thank Alina Ene and Huy Ngyuen
for pointing out to us their fast algorithm for merging two trees from \cite{en-19}.

\bibliographystyle{alpha}
\bibliography{refs}

\appendix
\input{appendix.tex}

\end{document}

%% file: intro.tex
\section{Introduction}
\label{sec:intro}
Spanning trees in graphs are a fundamental object of study and arise
in a number of settings. Efficient algorithms for finding a
minimum-cost spanning tree (MST) in a graph are classical. In a
variety of applications ranging from network design, TSP,
phylogenetics, and others, one often seeks to find a spanning tree
with additional constraints. An interesting and
well-known problem in this space is the
{\sc Bounded-Degree Spanning Tree} (\bdst) problem in which the goal
is to find a spanning tree in a given graph $G=(V,E)$ that minimizes
the maximum degree in the tree. We refer to the minimum-cost
version of \bdst as \bdmst where one seeks a spanning tree of minimum
cost subject to a given degree bound $B$ on the vertices. The decision
version of \bdst (Given $G,B$ is there a spanning tree with maximum
degree $B$?) is already NP-Complete for $B=2$ since it captures the
Hamilton-Path problem. In an influential paper, \Furer and Raghavachari
\cite{FurerR94}, building on earlier work of Win \cite{Win89},
described a simple local-search type algorithm that runs in $\tO(mn)$
time (here $m$ is number of edges and $n$ number of nodes) that
outputs a spanning tree with degree at most $B+1$, or certifies that
$G$ does not have a spanning tree with degree at most $B$ (we use
$\tO$ notation to suppress poly-logarithmic factors in $n,m,1/\eps$
for notational simplicity). Their
algorithm, in fact, works even in the non-uniform setting where each
vertex $v$ has a specified degree bound $B_v$. The \Furer-Raghavachari
result spurred a substantial line of work that sought to extend their
clean result to the minimum-cost setting. This was finally achieved by
Singh and Lau \cite{ls-15} who described a polynomial-time algorithm
that outputs a tree $T$ such that the degree of each $v$ in $T$ is at
most $B_v+1$ and the cost of the tree is at most $\OPT$. Their
algorithm is based on iterative rounding of a natural LP relaxation.
We refer the reader to \cite{KR05,Chaudhurietal07,Goemans06,ls-15,DuanHZ17}
for several ideas and pointers on \bdst and \bdmst.

Motivated by several applications, Bilo \etal \cite{BiloGRS04} defined
the {\sc Crossing Spanning Tree} problem (\cst).  In \cst the input is
a graph $G=(V,E)$, a collection of cuts $C_1,C_2,\ldots,C_k$, and
integers $B_1,B_2,\ldots,B_k$. Each cut $C_i$ is a subset of the edges
though in many applications we view $C_i$ as $\delta_G(S_i)$ for some
$S_i \subset V$ (where $\delta_G(S_i) = \{uv \in E \mid u \in S_i, v
\in V \setminus S_i\}$ is the standard definition of a cut set with
respect to $S_i$). The goal is to find a spanning tree $T$ such that
$|E(T) \cap C_i| \le B_i$, that is, $T$ crosses each cut $C_i$ at most
$B_i$ times. It is easy to see that \bdst is a special case of \cst
where the cuts correspond to singletons.  We refer to the min-cost
version of \cst as \cmst.  \cst gained substantial prominence in the
context of the asymmetric traveling salesman problem (ATSP) ---
Asadpour \etal \cite{Asadpouretal17} showed the importance of thin
spanning trees for approximating ATSP and obtained an $O(\log n/\log
\log n)$-approximation (now we have constant factor approximations for
ATSP via other methods \cite{SvenssonTV18,TraubV20}).  Motivated by
the thin tree conjecture and its applications to ATSP (see
\cite{Asadpouretal17,AnariG15}) and other technical considerations,
researchers have studied \cst, its generalization to the matroid
setting, and various special cases
\cite{cvz-10,BansalKN10,BansalKKNP13,OlverZ18,LinharesS18}.  The best
known approximation algorithms for \cst and its special cases have
mainly relied on the natural LP relaxation. For general \cst the best
know approximation ratio is $\min\{O(\log k/\log \log k), (1+\eps)B +
O(\log k/\eps^2)\}$.  A variety of sophisticated and interesting
rounding techniques have been designed for \cst and its special
cases. An outstanding open problem is whether \cst admits a constant
factor approximation via the natural LP relaxation. This is
challenging due its implications for the thin tree conjecture.

Most of the focus on \bdmst and \cst has been on the quality of the
approximation. The best known approximaton bounds rely on LP
relaxations and complex rounding procedures. The overall running times
are very large polynomials in the input size and are often
unspecified.  In this paper we are interested in the design of
\emph{fast} approximation algorithms for \bdmst, \cst and related
problems.  In recent years there has been significant progress in
designing fast, and often near-linear time, approximation algorithms
for a number of problems in discrete and combinatorial optimization.
This has been led by, and also motivated, synergy between
continuous/convex optimization, numerical linear algebra, dynamic data
structures, sparsification techniques, and structural results, among
several others. For \bdst with uniform degree, Duan, He and Zhang
\cite{DuanHZ17} described a combinatorial algorithm that for any given
$\eps > 0$, runs in $O(m\log^7 n/\eps^7)$ time, and either outputs a
spanning tree with degree $(1+\eps)B + O(\log n/\eps^2)$ or reports
that there does not exist a tree with maximum degree $\le B$.  This
paper is partly motivated by the goal of improving their results:
dependence on $\epsilon$, a better approximation, handling non-uniform
bounds, cost, \cmst, and connection to the LP relaxation.

A second motivation for this paper is to develop a fast algorithm for
swap-rounding in the spanning tree polytope. It is a
dependent rounding technique that has several applications ranging
from TSP to submodular function maximization (see
\cite{cvz-10,gw-17,cq-18b,en-19}). The question of developing a fast
swap-rounding procedure for spanning trees was explicitly raised in
\cite{cq-18b} in the context of Metric-TSP.

\subsection{Results}
In this paper we develop fast approximation algorithms for \bdmst,
\cmst and related problems in a unified fashion via broadly applicable
methodology based on the LP relaxation.  We consider
the following problem with general packing constraints. The input to
this problem is an undirected graph $G=(V,E)$, a non-negative
edge-cost vector $c: E \rightarrow \mathbb{R}_+$, a non-negative
matrix $A \in [0,1]^{k \times m}$, and a vector $b \in [1,\infty)^k$.
The goal is to find a spanning tree $T$ of minimum cost such that
$A\1_T \le b$ where $\1_T \in \{0,1\}^m$ is the characteristic vector
of the edge set of $T$. This is a special case of a more general
problem considered in \cite{cvz-10}: min-cost matroid base with
packing constraints.  Here we restrict attention to spanning trees
(graphic matroid). We refer to this slightly more general problem also
as \cmst.

Our first result is a near-linear time algorithm to approximately
solve the underlying LP relaxation for \cmst. For a multigraph $G$ we
let $\cT(G)$ denote the set of all spanning trees of $G$ and let
$\ST(G)$ denote the spanning tree polytope of $G$ (which is the convex
hull of the characteristic vectors $\{\1_T \mid T \in \cT(G)\}$).

\begin{theorem}
  \label{thm:intro-lp-solve}
  Let $G=(V,E)$ be a multigraph with $m$ edges and $n$ nodes and
  consider the linear program $\min \{ c^Tx : Ax \le b, x \in
  \ST(G)\}$ where $A \in [0,1]^{k\times m}, b \in [1,\infty)^k, c \in
  [0,\infty)^m$.  Let $N$ be the maximum  of $m$ and number of non-zeroes in $A$.
  There is a randomized polynomial time algorithm that for any given
  $\eps \in (0,1/2]$ runs in $\tO(N/\eps^2)$ time and with high
  probability either correctly certifies that the LP is infeasible or
  outputs a solution $y \in \ST(G)$ such that $c^Ty \le (1+\eps) \OPT$
  and $Ay \le (1+\eps)b$ where $\OPT$ is the minimum value of a
  feasible solution.
\end{theorem}

\begin{remark}
  We describe a randomized algorithm for the sake of simplicity,
  however we believe that a deterministic algorithm with similar
  guarantees can be obtained via ideas in \cite{cq-17}.
\end{remark}

Solving the LP relaxation quickly enables to estimate the optimum
integer solution \emph{value} via existing rounding results
\cite{ls-15,cvz-10,BansalKN10,BansalKKNP13,OlverZ18,LinharesS18}. For
instance, when specialized to \bdst, we obtain a near-linear time
algorithm to estimate the optimum value arbitrarily closely (modulo the
addditive $1$).

\begin{corollary}
  \label{cor:intro-bdmst}
  There is a randomized $\tO(m/\eps^2)$-time algorithm that outputs a
  value $B$ such that $B \le B^* \le \ceil{(1+\eps)B} + 1$ where $B^*$ is the
  minimum maximum degree over all spanning trees (that is, $B^* =
  \min_{T \in \cT(G)} \max_{v \in V} \deg_T(v)$ where $\deg_T(v)$ is
  the degree of $v$ in $T$).
\end{corollary}

Our second result shows the utility of the LP solution to sparsify the
original graph $G$.

\begin{theorem}
  \label{thm:intro-sparsification}
  Let $x \in \ST(G)$ be such that $Ax \le b$ for a
  matrix $A \in [0,1]^{k \times m}$ and $b \in [1,\infty)^k$. Consider
  a random subgraph $G'=(V,E')$ of $G$ obtained by picking each edge
  $e \in G$ with probability
  $\alpha_e := \min\{1,\frac{36\log (k+m)}{\eps^2}\cdot x_e\}$.  Then
  with high probability the following hold: (i)
  $\abs{E'} = O(n \ln (k+m)/\eps^2)$ (ii) there exists a fractional
  solution $z \in \ST(G)$ in the support of $G'$ such that
  $Az \le (1+3\eps)b$.
\end{theorem}

One can run a combinatorial algorithm such as the \Furer-Raghavchari
algorithm \cite{FurerR94} on the sparse graph rather than on the
original graph $G$.  This yields the following corollary which
improves the $\tilde{O}(mn)$
running time substantially when $G$ is dense.

\begin{corollary}
  \label{cor:intro-sparse-bdmst}
  There is a randomized algorithm for \bdst that given a graph $G$ on
  $n$ nodes runs in $\tO(n^2/\eps^2)$ time, and with high probability
  outputs a spanning tree $T$ with maximum degree
  $\ceil{(1+\eps)B^*} + 2$ where $B^*$ is the optimal degree bound.
\end{corollary}

\begin{remark}
  \label{remark:intro-non-uniform}
  Corollaries~\ref{cor:intro-bdmst} and \ref{cor:intro-sparse-bdmst}
  can be generalized to the non-uniform degree version of \bdst.
  Input is $G$ and degree bounds $B_v, v \in V$, and
  the algorithm either decides that there is no spanning  tree
  satisfying the degree bounds or outputs a tree that approximately
  satisfies them.
\end{remark}

Our final result is a fast algorithm to round the LP solution.
Several different rounding strategies have been developed for \bdmst
and \cmst and they yield different guarantees and take advantage of
the special structure of the given instance. Iterated rounding has
been one of the primary and powerful techniques, however it requires
basic feasible solutions to the LP relaxation; it seems far from
obvious how to obtain fast algorithms with comparable guarantees and
is a challenging open problem.  We are here interested in
\emph{oblivious randomized rounding} strategies that take a point $x
\in \ST(G)$ and round it to a random spanning tree $T \in \cT(G)$ such
that the coordinates of the resulting random edge vector are
\emph{negatively correlated}\footnote{A collection of $\{0,1\}$ random
  variables $X_1,X_2, \ldots, X_r$ are negatively correlated if, for all
  subsets $S \subseteq [r]$, $\E[\prod_{i \in S} X_i] \le \prod_{i \in
    S} \E[X_i]$ and $\E[\prod_{i \in S}(1-X_i)] \le \prod_{i \in S}
  (1-\E[X_i])$.}. Negative correlation implies concentration for linear
constraints as shown by Panconesi and Srinivasan \cite{PS97}.  These
strategies, when combined with the LP solution, yield bicriteria
approximation algorithms for \cmst of the form $(1+\eps, \min\{O(\log
k/\log \log k) b_i, (1+\eps)b_i + O(\log k)/\eps^2\})$ where the first
part is the approximation with respect to the cost and the second part
with respect to the packing constraints. For \cst and \cmst these are
currently the best known approximation ratios (although special cases
such as \bdmst admit much better bounds). Several dependent randomized
rounding techniques achieving negative correlation in the spanning
tree polytope are known: maximum entropy rounding
\cite{Asadpouretal17}, pipage rounding and \emph{swap rounding}
\cite{cvz-10}. These rounding techniques generally apply to matroids
and have several other applications.  In this paper we show that given
$x \in \ST(G)$, one can swap-round $x$ to a spanning tree in
near-linear time provided it is given in an implicit fashion;
alternately one can obtain an implicit \emph{approximate} representation $x'$
of $x$ and then apply an efficient swap-rounding on $x'$.  Since
swap-rounding is a flexible procedure and does not generate a unique
distribution, a precise technical statement requires more formal
notation and we refer the reader to Section~\ref{sec:swap}. Here we
state a theorem in a general form so that it can be used in other
contexts.

\begin{theorem}
  \label{thm:intro-swap-round}
  Let $G=(V,E)$ be a multigraph with $m$ edges and let
  $x \in [0,1]^m$. For any $\eps \in (0,1/2)$ there is a randomized
  algorithm that runs in $\tO(m/\eps^2)$ time and either correctly decides
  that $x \not \in \ST(G)$ or outputs a random vector
  $T = (X_1,X_2,\ldots,X_m) \in \{0,1\}^m$ such that (i) $T$ is the
  characteristic vector of a spanning tree of $G$ (ii)
  $\E[X_i] \le (1+\eps)x_i$ for $1 \le i \le m$ and (iii)
  $X_1,X_2,\ldots,X_m$ are negatively correlated. In particular $T$ is
  obtained as a swap-rounding of a vector $y$ such that $y \le (1+\eps)x$.
\end{theorem}

Combining Theorems~\ref{thm:intro-lp-solve} and
\ref{thm:intro-swap-round} and existing results on swap rounding
\cite{cvz-10} we obtain the following. The approximation ratio matches
the best known for \cmst and the algorithm runs in near-linear time.

\begin{corollary}\label{cor:intro-solve-round}
  For the feasibility version of \cmst, there is a randomized algorithm
  that runs in near-linear time and outputs a spanning tree $T$ such
  that
  \begin{align*}
    A\1_T \le \min\{O(\log k/\log \log k) b_i, (1+\eps)b_i + O(\log
    k)/\eps^2\}
  \end{align*}
  with high probability. For the cost version of
  \cmst, there is a randomized algorithm that outputs a
  \begin{align*}
    (1+\eps, \min\{O(\log k/\log \log k) b_i, (1+\eps)b_i + O(\log
    k)/\eps^2\})
  \end{align*}
  bicriteria approximation with probability $\Omega(\eps)$.
  After $\tO(1/\eps)$ independent repetitions of this algorithm, we
  can obtain the same guarantees with high probability.
\end{corollary}

Our algorithm, when specialized to \bdst and \bdmst is
more general than the one in \cite{DuanHZ17} in  terms  of  handling
cost and non-uniform degrees. In addition we obtain a very close
estimate of $B^*$, a much better dependence on $\eps$, and
also  obtain an approximation of the form $O(\log n/\log \log n) B^*$
which is better than $(1+\eps)B^* + O(\log n)/\eps^2$ for small $B^*$.

\paragraph{Extensions and other applications:} We mainly focused on
\bdmst and a high-level result for \cmst. One can obtain results for
related problems that involve multiple costs, lower bounds in addition
to upper bounds, and other applications of swap-roundings. We discuss
these in more detail in Section \ref{sec:putting-together}.

\subsection{Overview of main ideas}
Faster approximation algorithms for LPs that arise in combinatorial
optimization have been developed via several techniques. We follow a
recent line of work \cite{cq-17,cq-18,q-18,chq-20} that utilizes
features of the multiplicative weight update (MWU) method and data
structures to speed up \emph{implicit} LPs. In particular, the LP for
\cmst that we seek to solve can be addressed by the randomized MWU
algorithm from \cite{cq-18} and data structures for dynamic MST
\cite{hlt-01}. The overall approach follows some ideas from past work
\cite{cq-17}.  The sparsification result is inspired by recent
applications of similar ideas \cite{cq-17,cq-tsp,ChakLSSW19} and
utilizes Karger's theorem on random sampling for packing disjoint
bases in matroids \cite{k-98}.

Our main novel contribution is Theorem~\ref{thm:intro-swap-round}
which we believe is of independent interest beyond the applications
outlined here. Dependent randomized rounding techniques have had many
spectacular applications. In particular maximum entropy rounding in
the spanning tree polytope gave a strong impetus to this line of work
via its applications to ATSP \cite{Asadpouretal17} and metric-TSP
\cite{GharanSS11}. Swap-rounding is a simpler scheme to describe and
analyze, and suffices for several applications that only require
negative correlation. However, all the known dependent rounding
schemes are computationally expensive.  Recent work has led to
fantastic progress in sampling spanning trees \cite{AnariLOV20},
however the bottleneck for maximum entropy rounding is to compute,
from a given point $x \in \ST(G)$, the maximum entropy distribution
with marginals equal to $x$; polynomial time (approximation)
algorithms exist for this \cite{Asadpouretal17,SinghV14} but they are
rather slow. Swap-rounding \cite{cvz-10} requires one to decompose $x
\in \ST(G)$ (or more generally a point in the matroid base polytope)
into a convex combination of spanning trees; that is we write $x =
\sum_{T \in \cT} \lambda_T \1_T$ such that $\sum_T \lambda_T = 1$ and
$\lambda_T \ge 0, T \in \cT$. This is a non-trivial problem to do
exactly. The starting point here is a theorem in \cite{cq-17} that
shows that one can solve this decomposition problem
\emph{approximately} and deterministically in near-linear time via a
reduction to the problem of spanning tree packing; this is done via
MWU techniques. The near-linear time algorithm implies that any $x \in
\ST(G)$ can be decomposed efficiently into an \emph{implicit} convex
decomposition of total size $\tO(m/\eps^2)$ where $\eps$ is the
approximation parameter in the decomposition. We show in this paper
that this implicit sparse decomposition is well-suited to the
swap-rounding algorithm. We employ a divide-and-conquer strategy with
appropriate tree data structures to obtain an implementation that is
near-linear in the size of the implicit decomposition. Putting these
ingredients together yields our result.\footnote{In an earlier version of the
  paper (see \cite{cqt-20}) we described our fast swap rounding using two ideas.
  The first was a fast near-linear time algorithm to merge two
  spanning trees using the link-cut tree data structure.
  We were unaware of prior work of Ene and Nguyen
  \cite{en-19} that had already given such an algorithm
  in the context of fast algorithms for submodular function maximization
  in graphic matroids. In this version of the paper we use their
  algorithm as a black box. We focus on our second idea
  which exploits the implicit representation.
}

The seemingly fortuitous connection between the MWU based algorithm
for packing spanning trees and its implicit representation leading to
a fast algorithm for swap-rounding is yet another illustration of the
synergy between tools coming together in the design of
fast algorithms.

\subsection{Other related work}
We overview some known results on \cst and \cmst and special cases.
\bdmst can be viewed as a special case of \cmst where each edge
participates in $2$ constraints. Bansal \etal \cite{BansalKN10} showed
that if each edge participates in at most $\Delta$ constraints of $A$
(and $A$ is a binary matrix) then one can obtain a $(1, b + \Delta
-1)$-approximation generalizing the \bdmst result; this was further
extended to matroids by Lau, Kiraly and Singh \cite{KLS12}. It is
shown in \cite{BansalKKNP13} that for \cst one cannot obtain a purely
additive approximation better than $O(\sqrt{n})$ via the natural LP
relaxation. For this they use a reduction from discrepancy
minimization; it also implies, via the hardness result in
\cite{CharikarNN11} for discrepancy, that it is NP-Hard to obtain a
purely additive $o(\sqrt{n})$ bound.  Bansal \etal \cite{BansalKKNP13}
consider the \emph{laminar} case of \cmst where the cuts form a
laminar family and obtained a $(1, b+ O(\log n))$ approximation via
iterative rounding (this problem generalizes \bdmst). Olver and
Zenklusen \cite{OlverZ18} consider chain-constrained \cst which is a
further specialization when the laminar family is a chain (a nested
family of cuts). For this special case they obtained an $O(1)$-factor
approximation in the unit cost setting; Linhares and Swamy
\cite{LinharesS18} considered the min-cost version and obtained an
$(O(1),O(1))$-approximation. \cite{OlverZ18} also showed that even in
the setting of chain-constrained \cst, it is NP-Hard to
obtain a purely additive bound better than $c \log n/\log \log n$ for
some fixed constant $c$.

Dependent randomized rounding has been an active area of research with
many applications. Pipage rounding, originally devoped by Ageev and
Sviridenko \cite{pipage} in a deterministic way, was generalized to
the randomized setting by Srinivasan \cite{Srinivasan01} and by Gandhi
et al.~\cite{GandhiKPS06} and \cite{ccpv,cvz-10} and has led to a number
of applications. Maximum entropy rounding satisfies additional
properties beyond negative correlation and this is important in
applications to metric-TSP (see \cite{GharanSS11} and very recent work
\cite{KarlinKG20a,KarlinKG20b}).  There has been exciting recent
progress on sampling spanning trees and bases in matroids and we refer
the reader to some recent work \cite{Schild18,AnariGV18,AnariLOV20} for further
pointers.  Concentration bounds via dependent rounding can also be
obtained without negative correlation (see \cite{cvz-11} for instance)
and recent work of Bansal \cite{Bansal19} combines iterative rounding
with dependent rounding in a powerful way.

Fast approximation algorithms for solving positive LPs and SDPs has
been an extensive area of research starting from the early
90s. Lagrangean relaxation techniques based on MWU and other methods
have been extensively studied in the past, and continue to provide new
insights and results for both explicit and implicit problems. Recent
work based on a convex optimization perspective has led to a number of
new results and improvements. It is infeasible to do justice to this
extensive research area and we refer the reader to two recent PhD
theses \cite{q-thesis,w-thesis}. Spectacular advances in fast
algorithms based on the Laplacian paradigm, interior point methods,
cutting plane methods, spectral graph theory, and several others have
been made in the recent past and is a very active area of research
with frequent ongoing developments.

\paragraph{Organization:} Section~\ref{sec:prelim} introduces some
relevant notation, technical background and tree data structures
that we rely on. Section~\ref{sec:swap} describes our fast
swap-rounding algorithm and proves
Theorem~\ref{thm:intro-swap-round}.
Section~\ref{sec:sparsification} describes the sparsification process and proves
Theorem~\ref{thm:intro-sparsification}. Section~\ref{sec:solve} describes
our near-linear time algorithm to solve the LP relaxation for \cmst and proves
Theorem~\ref{thm:intro-lp-solve}. Section~\ref{sec:putting-together} brings
together results from previous sections to prove some of the corollaries stated
in the introduction and provides details of some extensions and related problems.


%% file: prelim.tex

\section{Preliminaries and notation}
\label{sec:prelim}
For a set $S$, we use the convenient notation $S-i$ to denote $S \setminus \{i\}$ 
and $S + i$ to denote $S \cup \{i\}$.

  \paragraph{Matroids.} We discuss some basics of matroids to establish some
  notation as well as present some useful lemmas that will be used later.
  A \emph{matroid} $\cM$ is a tuple $(N, \cI)$ with $\cI \subseteq 2^N$ satisfying
  the following three properties: (1) $\emptyset \in \cI$, (2) if $A \in \cI$ and 
  $B \subseteq A$, then $B \in \cI$, and (3) if $A,B \in \cI$ such that 
  $\abs{A} < \abs{B}$ then there exists $b \in B\setminus A$ such that 
  $A + b \in \cI$. We refer to the sets in $\cI$ as \emph{independent sets} and say
  that maximal independent sets are $\emph{bases}$. The \emph{rank} 
  of $\cM$ is the size of a base. For a set $A \in 2^N$, we refer to
  $r_\cM(A) = \max\{\abs{S} : S\subseteq A, S \in \cI\}$ as the \emph{rank}
  of $A$.
  
  A useful notion that we utilize in our fast implementation of swap rounding is
  that of contraction of a matroid.
  We say that the \emph{contraction} of $e$ in $\cM$ results in the
  matroid $\cM / e = (N - e, \{I \subseteq N -e : I + e \in \cI\})$ if 
  $r_\cM(\{e\}) = 1$ and $\cM/e = (N - e, \{I \subseteq N - e : I \in \cI\})$ if
  $r_\cM(\{e\}) = 0$. This definition extends naturally to contracting
  subsets $A \subseteq N$. It can be shown that contracting the
  elements of $A$ in any order results in the same matroid, which we
  denote as $\cM / A$.
  
  The following statements are standard results in the study of matroids (e.g.\ 
  see~\cite{Schrijver-book}).
  The following theorem is important in the analysis of swap rounding. It is often
  called the strong base exchange property of matroids.
  \begin{theorem}\label{thm:base}
    Let $\cM = (N, \cI)$ be a matroid and let $B,B'$ be bases.
    For $e \in B \setminus B'$, there exists $e' \in B' \setminus B$ such that
    $B - e + e' \in \cI$ and  $B' - e' +e \in \cI$.
  \end{theorem}

  The next lemma shows that if one contracts elements of an independent
  set in a matroid, bases in the contracted matroid can be used to form 
  bases in the initial matroid.
  \begin{lemma}\label{lem:cont-matroid-base}
    Let $\cM = (N,\cI)$ be a matroid and let $A \in \cI$.
    Let $B_A$ be a base in $\cM/A$. Then $A \cup B_A$ is a base in $\cM$.
  \end{lemma}

    \paragraph{A forest data structure.}
    We need a data structure to represent a forest that supports the necessary
    operations we need to implement randomized swap rounding in 
    Section~\ref{sec:swap}. The data structure mainly needs to facilitate
    the contraction of edges, including being able to recover the identity
    of the original edges after any number of contractions. 
    We enable this by enforcing that when the data structure is initialized, every edge
    $e$ is accompanied with a unique
    identifier. This identifier will be associated with the edge regardless of the
    edge's endpoints changing due to contraction. 
    The implementation of this step is important to guarantee a fast running time.
    The approach here is standard but requires some care, so we include the details
    in the appendix (see Section~\ref{app:forests-ds}).

%% file: swaprounding.tex

\section{Fast swap rounding in the spanning tree polytope}
\label{sec:swap}
Randomized swap rounding, developed in~\cite{cvz-10}, is a dependent
rounding scheme for rounding a fractional point $x$ in the base
polytope of a matroid to a random base $X$. The rounding preserves
expectation in that $\E[X] = x$, and more importantly, the coordinates
of $X$ are negatively correlated. 
In this section we prove
Theorem~\ref{thm:intro-swap-round} on a fast algorithm for
swap-rounding in the spanning tree polytope. We begin by describing
swap-rounding.

\subsection{Randomized swap rounding}
Let $\cM = (N, \cI)$ be a matroid and let $\cP$ be the base polytope of
$\cM$ (convex hull of the characteristic vectors of the bases of
$\cM$). Any $x \in \cP$ can be written as a finite convex combination
of bases: $x = \sum_{i=1}^h \delta_i \1_{B_i}$. Note that this
combination is not necessarily unique. As in~\cite{cvz-10}, we give
the original algorithm for randomized swap rounding via two routines.
The first is $\merge$, which takes as input two bases $B$, $B'$ and
two real values $\delta, \delta' \in (0,1)$. If $B = B'$ the algorithm
outputs $B$.  Otherwise the algorithm finds a pair of elements $e, e'$
such that $e \in B \setminus B'$ and $e' \in B' \setminus B$ where
$B - e + e' \in \cI$ and $B' - e' + e \in \cI$.  For such $e$ and
$e'$, we say that they are a \emph{valid exchange pair} and that we
\emph{swap} $e$ with $e'$.  The existence of such elements is
guaranteed by the strong base exchange property of matroids in
Theorem~\ref{thm:base}.  The algorithm randomly retains $e$ or $e'$ in
both bases with appropriate probability and this increases the
intersection size of $B$ and $B'$.  The algorithm repeats this process
until $B=B'$.  The overall algorithm $\swap$ utilizes $\merge$ as a
subroutine and repeatedly merges the bases until only one base is
left. A formal description is in Figure~\ref{alg:swap} along with the
pseudocode for $\merge$.

    \begin{figure}[t]
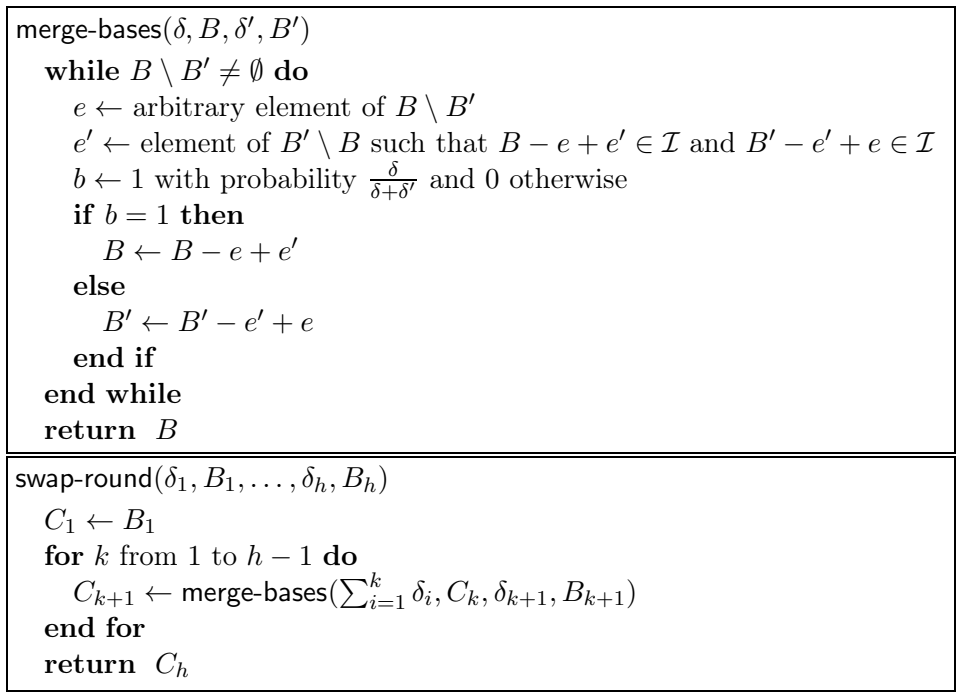

      \centering
      \fbox{\parbox{0.75\linewidth}
      {
        $\merge(\delta, B,\delta',B')$

        \begin{algorithmic}
          \WHILE{$B \setminus B' \ne \emptyset$}
            \STATE $e \leftarrow$ arbitrary element of $B \setminus B'$
            \STATE $e' \leftarrow$ element of $B' \setminus B$ such that
              $B - e + e' \in \cI$ and $B' - e' + e \in \cI$
            \STATE $b \leftarrow 1$ with probability $\frac{\delta}{\delta+\delta'}$ and
              $0$ otherwise
            \IF{$b=1$}
              \STATE $B \leftarrow B - e + e'$
            \ELSE
              \STATE $B' \leftarrow B' - e' + e$
            \ENDIF
          \ENDWHILE
          \RETURN{$B$}
        \end{algorithmic}
      }}
      \fbox{\parbox{0.75\linewidth}
      {
        $\swap(\delta_1,B_1,\ldots,\delta_h,B_h)$

        \begin{algorithmic}
          \STATE $C_1 \leftarrow B_1$
          \FOR{$k$ from $1$ to $h-1$}
            \STATE $C_{k+1} \leftarrow \merge(\sum_{i=1}^k \delta_i, C_k,
              \delta_{k+1}, B_{k+1})$
          \ENDFOR
          \RETURN{$C_h$}
        \end{algorithmic}
      }}
      \caption{The randomized swap rounding algorithm from~\cite{cvz-10}.}
      \label{alg:swap}
    \end{figure}

It is shown in \cite{cvz-10} that swap-rounding generates a random
base/extreme point $X \in \cP$ (note that the extreme points of $\cP$
are characteristic vectors of bases) such that $\E[X] = x$ and the
coordinates $X_1,X_2,\ldots,X_n$ (here $|N| = n$) are negatively
correlated. We observe that swap-rounding does not lead to a unique
probability distribution on the bases (that depends only $x$). First,
as we already noted, the convex decomposition of $x$ into bases is not
unique. Second, both $\merge$ and $\swap$ are non-deterministic in
their choices of which element pairs to swap and in which order to
merge bases in the convex decomposition. The key property for negative
correlation, as observed in \cite{cvz-10}, is to view the generation
of the final base $B$ as a vector-valued martingale (which preserves
expectation in each step) that changes only two coordinates in each
step. Another rounding strategy, namely pipage rounding, also enjoys
this property. Nevertheless swap-rounding is a meta algorithm that has
certain clearly defined features. The flexibility offered by $\merge$
and $\swap$ are precisely what allow for faster implementation in
specific settings.

We say that $\bar{B} \eqdist \merge(\delta, B, \delta',B')$ if $\bar{B}$ is the
random output of $\merge(\delta, B, \delta', B')$ for some
non-deterministic choice of valid exchange pairs in the algorithm.
Similarly we say that
$B \eqdist \swap(\delta_1,B_1,\ldots,\delta_h,B_h)$ if $B$ is the
random output of the $\swap$ process for some non-deterministic choice
of the order in which bases are merged and some non-deterministic
choices in the merging of bases. It follows from \cite{cvz-10} that
if $B \eqdist \swap(\delta_1,B_1,\ldots,\delta_h,B_h)$ then
$B$ satisfies the property that $\E[B] = x$ and coordinates of $B$
are negatively correlated.

\subsection{Setup for fast implementation in graphs}
Let $G=(V,E)$ be a multigraph with $|V| = n$ and $|E|=m$ and let
$x \in \ST(G)$ be a fractional spanning tree. Swap rounding requires
decomposing $x$ into a convex combination of spanning trees. This step
is itself non-trivial; existing algorithms have a high polynomial
dependence on $n,m$. Instead we will settle for an \emph{approximate}
decomposition that has some very useful features. We state
a theorem (in fact a corollary of a theorem) from \cite{cq-17} in
a slightly modified form suitable for us.

\begin{theorem}[Paraphrase of Corollary 1.2 in \cite{cq-17}]
  \label{thm:sparse-convex-decomp}
  Given a graph $G=(V,E)$ with $n=|V|$ and $m=|E|$ and a rational
  vector $x \in [0,1]^m$ there is a deterministic polynomial-time
  algorithm that runs in $\tO(m/\eps^2)$ time and either correctly
  reports that $x \not \in \ST(G)$ or outputs an \emph{implicit}
  convex decomposition of $z$ into spanning trees such that
  $z \le (1+\eps)x$.
\end{theorem}

The MWU algorithm behind the preceding theorem outputs a convex
decomposition of $z = \sum_{i=1}^h\delta_i \1_{T_i}$ for
$h = \tO(m/\eps^2)$ but in an implicit fashion. It outputs $T=T_1$ and
a sequence of tuples $(\delta_i, E_i,E_i')$ where
$T_{i+1} = T_i - E_i + E'_i$ for $1 \le i < h$ and has the property
that $\sum_{i=1}^{h-1} (|E_i| + |E'_i|) = \tO(m/\eps^2)$.  Thus the
convex decomposition of $z$ is rather sparse and near-linear in $m$
for any fixed $\eps > 0$. We will take advantage of this and
swap-round $z$ via this implicit convex decomposition. For many
applications of interest, including \cmst, the fact that we randomly
round $z$ instead of $x$ does not make much of a difference in the
overall approximation since $x$ itself in our setting
is the output of an approximate LP solver.

\begin{remark}
  The output of the approximate LP solver based on MWU for \cmst has
  the implicit decomposition as outlined in the preceding paragraph.
  However, for the sake of a self-contained result as stated in
  Theorem~\ref{thm:intro-swap-round}, we use the result from
  \cite{cq-17} which also has the advantage of being deterministic.
\end{remark}

The rest of the section describes a fast implementation for 
$\swap$. The algorithm is based on a divide and conquer
strategy for implementing $\swap$ when the convex combination is
described in an implicit and compact fashion. An important ingredient
is a fast black-box implementation of $\merge$. For this we use the
following result; as we remarked earlier, an earlier version of this paper
obtained a similar result.

\begin{theorem}[Ene and \Nguyen~\cite{en-19}]\label{thm:fast-merge}
  Let $T$ and $T'$ be spanning trees of a graph $G = (V,E)$
  with $\abs{V} = n$ and $E = T \cup T'$ and let $\delta,\delta' \in (0,1)$.
  There exists an algorithm $\fastmerge$ such that
  $\fastmerge(\delta,T,\delta',T') \eqdist \merge(\delta,T,\delta',T')$
  and the call to $\fastmerge$ runs in $O(n\log^2n)$ time. 
\end{theorem}

\subsection{Fast implementation of \texorpdfstring{$\swap$}{swap-round}}
In this subsection the goal is to prove the following theorem.

  \begin{theorem}\label{thm:fast-swap}
    Let $\sum_{i=1}^h \delta_i \1_{T_i}$ be a convex combination of spanning trees
    of the graph $G = (V,E)$ where $n = \abs{V}$.
    Let $T$ be a spanning tree such that $T= T_1$ and let $\{(E_i,E_i')\}_{i=1}^{h-1}$
    be a sequence of sets of edges such that
    $T_{i+1} = T_i - E_i + E_i'$ for all $i \in [h-1]$ and $E_i \cap E_i' = \emptyset$
    for all $i \in [h-1]$.
    Then there exists an algorithm that takes as input
    $T$, $\{(E_i,E_i')\}_{i=1}^{h-1}$, and $\{\delta_i\}_{i=1}^h$ and outputs a tree
    $T_S$ such that $T_S \eqdist \swap(\delta_1,T_1,\ldots,\delta_h,T_h)$.
    The running time of the algorithm is $\tO(n + \gamma)$ time where
    $\gamma = \sum_{i=1}^{h-1} (\abs{E_i} + \abs{E_i'})$.
  \end{theorem}


  \paragraph{A divide and conquer approach.} We consider the
  swap rounding framework in the setting of arbitrary matroids for
  simplicity. We work with the implicit decomposition of the
  convex combination of bases $\sum_{i=1}^h \delta_i\1_{B_i}$ of 
  the matroid $\cM = (N, \cI)$, as described in 
  Theorem~\ref{thm:fast-swap}. That is, the input is a base $B$ such
  that $B = B_1$, a sequence of sets of elements
  $\{(E_i,E_i')\}_{i=1}^{h-1}$ such that
  $B_{i+1} = B_i - E_i + E_i'$ and $E_i \cap E_i' = \emptyset$ for all
  $i \in [h-1]$, and the sequence of coefficients
  $\{\delta_i\}_{i=1}^h$. 
  
  The pseudocode for our divide and conquer algorithm $\divideswap$
  is given in Figure~\ref{alg:divide-swap}. The basic idea is simple.
  We \emph{imagine} constructing an explicit convex decomposition
  $B_1,B_2,\ldots,B_h$ from the implicit one. The high-level idea is to
  recursively apply swap rounding
  to $B_1,\ldots,B_{h/2}$ to create a base $B$, and similarly
  create a base $B'$ by recursively applying
  swap rounding to $B_{h/2+1},\ldots,B_h$, and then merging $B$ and $B'$.
  The advantage of this approach is manifested in the implicit case.
  To see this, we observe that in $\merge(\delta, B,\delta',B')$,
  the intersection $B \cap B'$ is always in the output, and this implies
  that the intersection $\bigcap_{i=1}^h B_i$ will always be in the
  output of $\swap(\delta_1,B_1,\ldots, \delta_h,B_h)$. Therefore, at
  every recursive level, we simply contract the intersection prior to
  merging any bases. Note that this is slightly complicated by the
  fact that the input is an implicit representation, but we note that
  Lemma~\ref{lem:set-2} implies $B \cap \bigcup_{i=1}^{h-1} (E_i\cup
  E_i')= B \setminus \bigcap_{i=1}^{h} B_i$ (see proof of
  Lemma~\ref{lem:divide-swap}). (We note later how the contraction of
  elements helps in the running time when specializing to the graphic
  matroid.)  After contracting the intersection, the algorithm
  recursively calls itself on the first $h/2$ bases and the second
  $h/2$ bases, then merges the output of the two recursive calls via
  $\merge$. With the given implicit representation, this means that
  the input to the first recursive call is $B_1,
  \{(E_i,E_i')\}_{i=1}^{h/2-1},\{\delta_i\}_{i=1}^{h/2}$ and the input
  to the second recursive call is $B_{h/2+1},
  \{(E_i,E_i')\}_{i=h/2+1}^{h-1},\{\delta_i\}_{i=h/2+1}^h$ (note we
  can easily construct $B_{h/2+1}$ via the implicit representation).
  The underlying matroid in the call to $\merge$ is the matroid $\cM$
  with the intersection $\bigcap_{i=1}^h B_i$ contracted.

  \begin{figure}[t]
    \centering
    \fbox{\parbox{0.66\linewidth}
    {
      $\divideswap(B, \{(E_i,E_i')\}_{i=s}^{t-1},\{\delta_i\}_{i=s}^t)$

      \begin{algorithmic}
        \IF{$s =t$}
          \RETURN{$B$}
        \ENDIF
        \STATE $\ell \leftarrow \max\left\{\ell' \in \{s, s+1,\ldots, t\} :
          \sum_{i=s}^{\ell'-1}\abs{E_i} \le \frac{1}{2} \sum_{i=s}^{t-1} \abs{E_i}\right\}$
        \STATE $\hat{B} \gets B \cap \bigcup_{i=s}^{t-1}(E_i \cup E_i')$
        \STATE $\hat{B}_C \gets \hat{B}$
        \FOR{$i$ from $s$ to $\ell$}
          \STATE $\hat{B}_C \gets \hat{B}_C - E_i + E_i'$
        \ENDFOR
        \STATE $\hat{B}_L \gets
          \divideswap(\hat{B}, \{(E_i,E_i')\}_{i=s}^{\ell-1}, \{\delta_i\}_{i=s}^\ell)$
        \STATE $\hat{B}_R\gets \divideswap(\hat{B}_C,
          \{(E_i,E_i')\}_{i=\ell+1}^{t-1}, \{\delta_i\}_{i=\ell+1}^t)$
        \STATE $\hat{B}_M \gets \merge(\sum_{i=s}^\ell \delta_i, \hat{B}_L,
          \sum_{i=\ell+1}^t \delta_i, \hat{B}_R)$
        \RETURN{$\hat{B}_M \cup (B \setminus \bigcup_{i=s}^{t-1} (E_i \cup E_i'))$}
      \end{algorithmic}
    }}
    \caption{A divide-and-conquer implementation of swap rounding with
      an implicit representation.}
    \label{alg:divide-swap}
  \end{figure}

  The following lemma shows that $\divideswap$ is a proper implementation
  of $\swap$. We leave the details of the proof for the appendix 
  (see Section~\ref{app:proofs}).
  \begin{lemma}\label{lem:divide-swap}
    Let $\sum_{i=1}^h \delta_i\1_{B_i}$ be a convex combination of bases in
    the matroid $\cM$ and let $\{(E_i,E_i')\}_{i=1}^{h-1}$ be a sequence of elements
    such that $B_{i+1} = B_i - E_i + E_i'$ and $E_i \cap E_i' = \emptyset$ for all
    $i \in [h-1]$. Then
    $ \swap(\delta_1,B_1,\ldots,\delta_h,B_h) \eqdist
    \divideswap(B_1,\{(E_i,E_i')\}_{i=1}^{h-1}, \{\delta_i\}_{i=1}^h)$.
  \end{lemma}

  \paragraph{A fast implementation of $\divideswap$ for spanning trees.}
  The pseudocode for our fast implementation $\fastswap$ of 
  $\divideswap$ is given in Figure~\ref{alg:fast-swap}. 
  
  As in $\divideswap$, the algorithm $\fastswap$ contracts the intersection of 
  the input. Suppose we contract the intersection  
  $\bigcap_{i=1}^h T_i$ in $T_j$ and call this contracted tree 
  $\hat{T}_j$. Then $\abs{\hat{T}_j} 
  \le \abs{T_j \setminus \bigcap_{i=1}^h T_i}$.
  By Lemma~\ref{lem:set-1}, we have
  $T_j \setminus \bigcap_{i=1}^h T_i \subseteq
  \bigcup_{i=1}^{h-1} (T_i \setdiff T_{i+1}) = \bigcup_{i=1}^{h-1}
  (E_i \cup E_i')$ (note $A \setdiff B$ is the symmetric difference of the
  sets $A,B$). Thus, the size of the contracted tree is bounded 
  by the size of the implicit representation 
  $\gamma := \sum_{i=1}^{h-1}\abs{E_i}
  +\abs{E_i'}$. 
  One can write a convex combination of bases in any matroid using 
  the implicit representation, and contraction could even be implemented 
  quickly as is the case in the graphic matroid. The
  main point for improving the running time is having an implementation
  of $\merge$ that runs in time proportional to the size of the
  \emph{contracted} matroid. This is key to the speedup achieved for the 
  graphic matroid. $\fastmerge$ runs in time proportional to the 
  size of the input trees, which have been contracted to have size 
  $O(\min\{n,\gamma\})$, which yields
  a running time of $\tO(\min\{n,\gamma\})$. This speedup at every recursive
  level combined with the divide-and-conquer approach of $\fastswap$ 
  is sufficient to achieve a near-linear time implementation of $\swap$.

  Recall that as we are working with contracted trees, an edge in the 
  contracted trees might have different endpoints than it did in the initial 
  trees. The identifiers
  of edges do not change, regardless of whether the endpoints of the edge
  change due to contraction of edges in a tree. We therefore will refer to id's
  of edges throughout the algorithm $\fastswap$ to work from contracted
  edges back to edges in the initial trees. This extra 
  bookkeeping will mainly be handled implicitly.

  Contraction of the intersection of the input trees in $\fastswap$ using
  only the implicit representation is handled by the
  the algorithm $\shrinkint$ and we give the pseudocode in
  Figure~\ref{alg:shrink-int}.  Consider spanning trees
  $T_s,T_{s+1},\ldots, T_t$. The input to $\shrinkint$ is $T_s$ and a sequence
  of sets of edges $\{(E_i,E_i')\}_{i=s}^{t-1}$ such that
  $T_{i+1} = T_i - E_i + E_i'$ and $E_i \cap E_i' = \emptyset$ for
  $i \in \{s,s+1,\ldots, t-1\}$. Then $\shrinkint$ contracts
  $\bigcap_{i=s}^t T_i$ in $T_s$. The intersection is computed using
  Lemma~\ref{lem:set-2}, which shows that
  $\bigcap_{i=s}^t T_i = T_s \setminus \bigcup_{i=s}^{t-1} (E_i \cup
  E_i')$.  Let $\hat{T}_s$ be $T_s$ with $\bigcap_{i=s}^t T_i$
  contracted. The vertex set of $\hat{T}_s$ is different than the
  vertex set of $T_s$. Then as the sets of edges $E_i$ and $E_i'$ for
  all $i$ are defined on the vertices in $T_s$, we need to map the
  endpoints of edges in $E_i$ to the new vertex set of
  $\hat{T}_s$. Using the data structure presented in
  Lemma~\ref{lem:forest}, this is achieved using the operations
  $\represented$ and $\origedge$, which handle mapping the edge to its
  new endpoints and maintaining the edge identifier, respectively.
  
  The following lemma shows that $\shrinkint$ indeed contracts the
  intersection of the trees via the implicit representation.
  The proof is in the appendix (see Section~\ref{app:proofs}).

  \begin{figure}[t]
    \centering
    \fbox{\parbox{0.55\linewidth}
    {
      $\shrinkint(T, \{(E_i,E_i')\}_{i=s}^{t-1})$

      \begin{algorithmic}
        \STATE $\hat{T} \gets T.\copyy()$
        \FOR{$e \in T \setminus \bigcup_{i=s}^{t-1}(E_i \cup E_i')$}
          \STATE $uv \gets \hat{T}.\represented(e)$
          \STATE assume
            $\deg_{\hat{T}}(u) \le \deg_{\hat{T}}(v)$
            (otherwise rename)
          \STATE $\hat{T}.\contract(uv, u)$
        \ENDFOR
        \STATE let $id(e)$ denote the unique identifier of an edge $e$
        \FOR{$i$ from $s$ to $t-1$}
          \STATE $\hat{E}_i \gets \bigcup_{e \in E_i} (\hat{T}.\represented(e),id(e))$
          \STATE $\hat{E}_i' \gets \bigcup_{e \in E_i'} (\hat{T}.\represented(e),id(e))$
        \ENDFOR
        \RETURN{$(\hat{T}, \{(\hat{E}_i,\hat{E}_i')\}_{i=s}^{t-1})$}
      \end{algorithmic}
    }}
    \caption{A subroutine used in our fast implementation $\fastswap$ of randomized
    swap rounding; used to implicitly contract the trees of the given convex combination.}
    \label{alg:shrink-int}
  \end{figure}
  \begin{lemma}\label{lem:shrink-int}
    Let $T_1, \ldots, T_h$ be spanning trees and let $\{(E_i,E_i')\}_{i=1}^{h-1}$
    be a sequence of edge sets
    defined on the same vertex set such that $T_{i+1} = T_i - E_i + E_i'$
    and $E_i \cap E_i' = \emptyset$ for all
    $i \in [h-1]$. Contract $\bigcap_{i=1}^h T_i$ in $T_1,\ldots,T_h$
    to obtain $\hat{T}_1,\ldots,\hat{T}_h$, respectively.

    Let $n_{T_1} = \abs{T_1}$ and
    $\gamma = \sum_{i=1}^{h-1} (\abs{E_i} + \abs{E_i'})$.
    Then $\shrinkint(T_1,\{(E_i,E_i')\}_{i=1}^{h-1})$ runs in time
    $\tO(n_{T_1} + \gamma)$
    and outputs $(\hat{T}, \{(\hat{E}_i, \hat{E}_i')\}_{i=1}^{h-1})$ where
    $\hat{T} = \hat{T}_1$ and $\hat{T}_{i+1} = \hat{T}_i - \hat{E}_i + \hat{E}_i'$
    for all $i \in [h-1]$. Moreover, $\abs{E_i} = \abs{\hat{E}_i}$
    and $\abs{E_i'} = \abs{\hat{E}_i'}$ for all $i \in [h-1]$ and
    $\abs{\hat{T}} \le \min\{n_{T_1},\gamma\}$.
  \end{lemma}

  \begin{figure}[t]
    \centering
    \fbox{\parbox{0.65\linewidth}
    {
      $\fastswap(T, \{(E_i,E_i')\}_{i=s}^{t-1},\{\delta_i\}_{i=s}^t)$
      \begin{algorithmic}
        \IF{$s=t$}
          \RETURN{$T$}
        \ENDIF
        \STATE $\ell \leftarrow \max\left\{\ell' \in \{s, s+1,\ldots, t\} :
          \sum_{i=s}^{\ell'-1}\abs{E_i} \le \frac{1}{2} 
          \sum_{i=s}^{t-1} \abs{E_i}\right\}$
        \STATE $(\hat{T}, \{(\hat{E}_i, \hat{E}_i')\}_{i=s}^{t-1})
          \gets\shrinkint(T, \{(E_i,E_i')\}_{i=s}^{t-1})$
        \STATE $\hat{E}_C \gets E(\hat{T})$
        \FOR{$i$ from $s$ to $\ell$}
          \STATE $\hat{E}_C \gets \hat{E}_C - \hat{E}_i + \hat{E}_i'$
        \ENDFOR
        \STATE $\hat{T}_C \gets \init(V(\hat{T}), \hat{E}_C)$
        \STATE $\hat{T}_L \gets
          \fastswap(\hat{T}, \{(\hat{E}_i,\hat{E}_i')\}_{i=s}^{\ell-1}, 
          \{\delta_i\}_{i=s}^\ell)$
        \STATE $\hat{T}_R\gets\fastswap(\hat{T}_C,
          \{(\hat{E}_i,\hat{E}_i')\}_{i=\ell+1}^{t-1}, \{\delta_i\}_{i=\ell+1}^t)$
        \STATE $\hat{T}_M \gets \fastmerge(\sum_{i=s}^\ell \delta_i, \hat{T}_L,
          \sum_{i=\ell+1}^t \delta_i, \hat{T}_R)$
        \RETURN{$\hat{T}_M \cup
          (T \setminus \bigcup_{i=s}^{t-1} (E_i \cup E_i'))$}
      \end{algorithmic}
    }}
    \caption{A fast implementation of randomized swap rounding from~\cite{cvz-10}.}
    \label{alg:fast-swap}
  \end{figure}

  We use the algorithm in Theorem~\ref{thm:fast-merge} for
  $\merge$. In $\fastswap$, the two trees that are merged $\hat{T}_L$
  and $\hat{T}_R$ are the return values of the two recursive calls to
  $\fastswap$. The algorithm at this point has explicit access to the
  adjacency lists of both $\hat{T}_L$ and $\hat{T}_R$, which are used
  as input to the algorithm $\fastmerge$. The output of $\fastmerge$
  will be the outcome of merging the two trees $\hat{T}_L$ and
  $\hat{T}_R$, which are edges of potentially contracted trees from
  the original convex combination.  We can use the operation
  $\origedge$ of the forest data structure of Lemma~\ref{lem:forest}
  for $\hat{T}_L$ and $\hat{T}_R$ to obtain the edges from the trees
  of the original convex combination. This extra bookkeeping will be
  handled implicitly.
  
  We next prove that $\fastswap$ is implementing $\swap$ and that
  it runs in near-linear time.
  \begin{lemma}\label{lem:fast-swap-run}
    Let $\sum_{i=1}^h \delta_i \1_{T_i}$ be a convex combination of spanning trees
    of the graph $G = (V,E)$ where $n = \abs{V}$.
    Let $T$ be a spanning tree such that $T= T_1$ and let $\{(E_i,E_i')\}_{i=1}^{h-1}$
    be a sequence of sets of edges such that
    $T_{i+1} = T_i - E_i + E_i'$ and $E_i \cap E_i' = \emptyset$ for all $i \in [h-1]$.

    Then $\fastswap(T, \{(E_i,E_i')\}_{i=1}^{h-1},
    \{\delta_i\}_{i=1}^h) \eqdist \swap(\delta_1,T_1,\ldots,\delta_h,T_h)$
    and the call to $\fastswap$ runs in $\tO(n_T + \gamma)$ time where
    $n_T = \abs{T}$ and
    $\gamma = \sum_{i=1}^{h-1} (\abs{E_i} + \abs{E_i'})$.
  \end{lemma}
  \begin{proof}
    One can immediately see that
    $\fastswap$ is an implementation of $\divideswap$. There are some 
    bookkeeping details that are left out of the implementation, such as maintaining
    the set of edges returned by $\fastmerge$, but these are easily handled.
    Lemma~\ref{lem:divide-swap} shows that
    $\divideswap(\delta_1,T_1,\ldots,\delta_h,T_h) \eqdist
    \swap(\delta_1,T_1,\ldots,\delta_h,T_h)$, implying
    $\fastswap(T, \{(E_i,E_i')\}_{i=1}^{h-1},
    \{\delta_i\}_{i=1}^h) \eqdist \swap(\delta_1,T_1,\ldots,\delta_h,T_h)$.

    Now we prove the running time bound.  Let $R(n_T, \gamma)$ denote
    the running time of the call to
    $\fastswap(T,\{(E_i,E_i')\}_{i=1}^{h-1},\{\delta_i\}_{i=1}^h)$.
    By Lemma~\ref{lem:shrink-int}, the running time of the call to
    $\shrinkint$ is $\tO(n_T + \gamma)$.  Let
    $(\hat{T}, \{(\hat{E}_i,\hat{E}_i')\}_{i=1}^{h-1})$ be the output
    of $\shrinkint$. Lemma~\ref{lem:shrink-int} also guarantees that
    $\gamma = \sum_{i=1}^{h-1} \abs{\hat{E}_i} + \abs{\hat{E}_i'}$ and
    $\abs{\hat{T}} \le \min\{n_T,\gamma\}$. Then by
    Lemma~\ref{lem:forest}, constructing $\hat{T}_C$ requires
    $\tO(n_{\hat{T}} + \gamma)$ time.  As the size of $\hat{T}_L$ and
    $\hat{T}_R$ is the same as $\hat{T}$ and $\hat{T}_C$, 
    the call to $\fastmerge$ runs in $\tO(n_T + \gamma)$ time by
    Theorem~\ref{thm:fast-merge}. The time it takes to compute the
    returned tree is $\tO(n_T + \gamma)$ as we have enough time to scan
    all of $T$ and $\bigcup_{i=1}^{h-1}(E_i \cup E_i')$.  So the total
    time excluding the recursive calls is $(n_T + \gamma)\cdot \alpha$ for
    where $\alpha = O(\log^c (n_T + \gamma))$ for some fixed integer $c$.

    As for the recursive calls, first define $\gamma(s,t) := \sum_{i=s}^t (\abs{E_i}
    + \abs{E_i'})$. Then the running time of the first recursive call is
    $R(n_{\hat{T}}, \gamma(1,\ell-1))$ and the second recursive call
    is $R(n_{\hat{T}_C},\gamma(\ell+1,h-1))$.By choice of $\ell$, we always have
    that $\gamma(1,\ell-1) \le \frac{\gamma}{2}$. As $\ell$ is the largest
    integer such that $\gamma(1,\ell-1) \le \frac{\gamma}{2}$, then
    $\gamma(1,\ell) > \frac{\gamma}{2}$. Therefore, we have
    $\gamma(\ell+1, h-1) = \gamma - \gamma(1,\ell) < \frac{\gamma}{2}$. Combining
    this with the fact that $\shrinkint$ guarantees that
    $\abs{\hat{T}} \le \min\{n_T,\gamma\}$ and $\abs{\hat{T}_C} \le \min\{n_T,\gamma\}$, we have
    \[
      R(n_T,\gamma)
      \le
      2R(\min\{n_T, \gamma\}, \gamma/2) + \alpha \cdot (n_T + \gamma).
    \]
    Note that $R(n_T,\gamma) = O(1)$ when $n_T = O(1)$ and $\gamma = O(1)$.

    We claim that $R(a,b) \le \alpha \beta \cdot (a + 8b\log b) $ is a valid
    solution to this recurrence for some sufficiently large but fixed
    constant $\beta \ge 1$. By choosing $\beta$ sufficiently large
    it is clear that it holds for the base case. To prove the
    inductive step we see the following:
    \[
      R(a,b)
      \le
      2R(\min\{a,b\}, b/2) + \alpha \cdot (a + b)
      \le
      2[\alpha \beta \cdot (\min\{a,b\}  + 4b\log(b/2))]
      + \alpha \cdot (a + b).
    \]
    Hence we need to verify that
    \begin{equation}\label{eq:rec}
      2[\alpha \beta (\min\{a,b\} + 4b\log(b/2))] + \alpha\cdot (a + b)
      \le
      \alpha\beta \cdot (a + 8b \log b).
    \end{equation}
    Since $\beta \ge 1$, rearranging, it suffices to verify that
    \[
      2\min\{a,b\} + 8b \log (b/2) + b \le 8b\log b.
    \]
    As $8b\log b - 8b\log(b/2) = 8b$ and $2\min\{a,b\} + b\le3b$, this
    proves~(\ref{eq:rec}) and therefore
    $R(n_T, \gamma) \le \alpha \beta (n_T + 8\gamma \log \gamma) = \tO(n_T + \gamma)$.
    This concludes the proof.
  \end{proof}

\paragraph{Proof of Theorem~\ref{thm:intro-swap-round}:} Follows by
combining Theorem~\ref{thm:sparse-convex-decomp} (and remarks after the
theorem statement) and Theorem~\ref{thm:fast-swap}.

%% file: sparsification.tex

\section{Sparsification via the LP Solution}
\label{sec:sparsification}
\newcommand{\bx}{x} 
\newcommand{\by}{y}
\newcommand{\bz}{z}

Let $G=(V,E)$ be a graph on $n$ nodes and
$m$ edges and let $\bx$ be a point in $\ST(G)$. In this section we
show how one can, via random sampling, obtain a sparse point
$\bx' \in \ST(G)$ from $\bx$. The random sampling approximately
preserves linear constraints and thus one can use this technique to
obtain sparse LP solutions to the packing problems involving
spanning tree (and more generally matroid) constraints.  The sampling
and analysis rely on Karger's well-known work on random sampling for
packing disjoint bases in matroids. We paraphrase the relevant
theorem.

\begin{theorem}[Corollary 4.12 from \cite{k-98}]
  \label{thm:matroid-sampling}
  Let $\cM$ be a matroid with $m$ elements and non-negative
  integer capacities on elements such
  that $\cM$ has $k$ disjoint bases. Suppose each copy of
  a capacitated element $e$ is sampled independently with
  probability $p \ge 18 (\ln m)/(k\eps^2)$ yielding a matroid
  $\cM(p)$. Then with high probability the  number of
  disjoint bases in $\cM(p)$ is in $[(1-\eps)pk, (1+\eps)pk]$.
\end{theorem}

We now prove Theorem~\ref{thm:intro-sparsification}. We restate
it here for the reader's convenience. 
\begin{manualtheorem}{1.4}
  Let $\bx \in \ST(G)$ be a rational vector such that $A\bx \le b$ for
  a matrix $A \in [0,1]^{r \times m}$ and $b \in [1,\infty)^r$. Consider a random
  subgraph $G'=(V,E')$ of $G$ obtained by picking each edge $e \in G$
  with probability
  $\alpha_e := \min\{1,\frac{36\log (r+m)}{\eps^2}\cdot x_e\}$.  Then
  with high probability the following hold: (i)
  $\abs{E'} = O(n \ln (r+m)/\eps^2)$ (ii) there exists a fractional
  solution $\bz \in \ST(G)$ in the support of $G'$ such that $A\bz \le (1+3\eps)b$.
\end{manualtheorem}

\begin{proof}
  We will consider two random subgraphs of $G$ and then
  relate them.

  Let $\beta = \frac{18 \log (m+r)}{\eps^2}$.  Since $\bx$ is rational
  there exists a positive integer $L \ge \beta$ such that $Lx_e$ is an
  integer for all $e \in E$. Let $\hat{G}=(V,\hat{E})$ be a multigraph
  obtained from the graph $G=(V,E)$ as follows: we replace each edge
  $e \in G$ with $Lx_e$ parallel edges. Let $P(\hat{G})$ be the
  maximum number of edge-disjoint spanning trees in $\hat{G}$. Since
  $\bx \in \ST(G)$ and $L\bx$ is an integer vector, by the integer
  decomposition property of the matroid base polytope
  \cite{Schrijver-book}, $P(\hat{G}) = L$. Let $H$ be the graph
  resulting from sampling each edge in $\hat{G}$ independently
  with probability $p$, where $pL = \beta$.  By
  Theorem~\ref{thm:matroid-sampling}, with high probability, we have
  $P(H) \ge (1-\eps) \beta$.

  For all $e \in E(G)$, let $\ell(e)$ be the number of parallel copies
  of $e$ in the random graph $H$ and define $y_e := \min\{1,
  \frac{1+\eps}{\beta} \ell(e)\}$.  Since $P(H) \ge (1-\eps)\beta$. we
  see that ${\bf y}$ dominates a convex combination of spanning trees
  of $G$, and hence ${\bf y}$ is in the dominant of $\ST(G)$. This
  implies that there is a ${\bf z} \in \ST(G)$ such that ${\bf z} \le
  {\bf y}$.

  Consider any row $i$ of $A$ where $1 \le i \le r$. We have
  $\sum_{e} A_{i,e} x_e \le b_i$.  Consider the random variable
  $Q = \sum_e A_{i,e} \ell(e)$.  We have
  $\E[Q] = \beta \sum_{e} A_{i,e} x_e \le \beta b_i$ and moreover $Q$
  is the sum of binary random variables (recall that we pick each copy
  of $e$ in $H$ independently with probability $p$).  Hence we can
  apply Chernoff-Hoeffding bounds to $Q$ (note that
  $A \in [0,1]^{r \times m}$) to conclude that
  $\Pr[Q \ge (1+\eps) \beta b_i] \le \exp(-\eps^2 \beta b_i/2)$. Since
  $b_i \ge 1$ and $\beta \ge 18 \ln (m+r)/\eps^2$ we have that this
  probability is inverse polynomial in $r,m$.  Note that
  $\sum_{e}A_{i,e} y_e \le \frac{1+\eps}{\beta} \sum_e A_{i,e}
  \ell(e)$. Hence
  $\sum_e A_{i,e} y_e \le (1+\eps)^2 b_i \le (1+3\eps)b$ with high
  probability.  By union bound over all $r$ constraints we conclude
  that, with high probability,
  $A {\bf z} \le A{\bf y} \le (1+3\eps) b$ holds.

  Let $H'$ be a random subgraph of $G$ consisting of all edges $e$
  such that $\ell(e) > 0$. We can view $H'$ as being obtained via a
  process that selects each edge $e \in G$ independently with
  probability $\rho_e$. The argument above shows that with high
  probability there is ${\bf z} \in \ST(G)$ such that
  $A{\bf z} \le (1+3\eps)b$ and ${\bf z}$ is in the support of $H'$.
  Now consider the process of selecting each edge $e$ independently
  with probability $\alpha_e$. If we show that $\alpha_e \ge \rho_e$ for
  every $e$ then, by stochastic dominance, it follows that with high
  probabilty the resulting random graph will also contain a ${\bf z}$
  with the desired properties which is what we desire to prove.
  Note that the probability an edge $e$ being in $H'$ is
  $\rho_e = 1 - (1 - p)^{Lx_e}$. If $x_e > 1/2$ we have  $\alpha_e = 1
  \ge \rho_e$ so we assume $x_e \le 1/2$.
    \[
      \rho_e = 1 - (1-p)^{Lx_e}
      \le
      1 - e^{-2Lx_ep}
      \le
      1 - (1 - 2Lx_ep)
      =
      2\beta x_e
      \le
      \alpha_e,
    \]
    where the first inequality holds because $1 - x \ge e^{-2x}$ for
    $x \in [0,\frac{1}{2}]$.

    All that remains is to show that the number of edges produced by
    selecting each edge $e$ independently with $\alpha_e$ is
    $O(n \ln (r+m)/\eps^2)$. This holds with high probability via
    Chernoff-Hoeffding bounds since $\sum_e x_e = n-1$.

    By union bound we have that with high probability the random process 
    produces a graph that has $O(n \ln (r+m)/\eps^2)$ edges
    and supports a fractional solution ${\bf z}  \in \ST(G)$ such
    that $A{\bf z} \le (1+3\eps)b$.
  \end{proof}

  \begin{remark}
    For problems that also involve costs, we have a
    fractional solution ${\bf x}$ and an objective $\sum_{e} c_e x_e$.
    Without loss of generality we can assume that $c_e \in [0,1]$ for
    all $e$. The preceding proof shows that the sparse graph obtained by sampling
    supports a fractional solution ${\bf z}$ such that
    $\E[\sum_e c_e z_e] \le (1+\eps)\sum_e c_e x_e$. Further, $\sum_e c_e z_e \le (1+3\eps)\sum_e c_e x_e$ holds with high probability as long as
    $\max_e c_e \le \sum_e c_e x_e$.  This condition may not hold in general but
    can be typically guaranteed in the overall algorithm by guessing
    the largest cost edge in an optimum integer solution.
  \end{remark}
  
\begin{remark}
  The proof in the preceding theorem also shows that with high probability the graph $G'$ is connected and satisfies the property that $A\1_{E'} \le O(\log n/\eps^2) b$. Thus any spanning tree of $G'$ satisfies the constraints to a multiplicative $O(\log n)$-factor by fixing $\epsilon$ to a small constant. This is weaker than the guarantee provided by swap-rounding.
\end{remark}

%% file: lp-solve.tex

\section{Fast approximation scheme for solving the LP
  relaxation}\label{sec:solve}
In this section we describe a fast approximation scheme to solve the
underlying LP relaxation for \cst and prove
Theorem~\ref{thm:intro-lp-solve}. We recall the LP for \cst.

  \begin{equation}\label{pr:compact}\tag{$\cP$}
    \begin{array}{l l l}
      \min &  \sum_{e \in E} c_e y_e &\\
      \text{subject to} & Ay \le b &\\
      &  y \in \ST(G) &
    \end{array}
  \end{equation}

  Note that even the feasibility problem (whether there exists
  $y \in \ST(G)$ such that $Ay \le b$) is interesting and
  important. We will mainly focus on the feasibility LP and show
  how we can incorporate costs in Section~\ref{sec:costs}.  We recast
  the feasibility LP as a pure packing LP using an implicit
  formulation with an exponential number of variables. This is
  for technical convenience and to more directly apply the framework
  from \cite{cq-18}. For each
  tree $T \in \cT(G)$ we have a variable $x_T$ and we consider the
  problem of packing spanning trees.

  \begin{equation}\label{pr:main}\tag{$\cC$}
    \begin{array}{l l l}
    \text{maximize }  & \displaystyle \sum_{T \in \cT} x_T &\\
    \text{subject to}
      &\displaystyle \sum_{T \in \cT} (A\1_T)_i\cdot x_T \le b_i,
        & \forall i \in [k]\\
      & x_T \ge 0,
        &\forall T\in \cT
    \end{array}
  \end{equation}

  The following is easy to establish from the fact that $y \in \ST(G)$
  iff $y$ can be written as a convex combinaton of the characteristic
  vectors of spanning trees of $G$.
  \begin{observation}
    There exists a feasible solution $y$ to \ref{pr:compact} iff there
    exists a feasible solution $x$ to \ref{pr:main} with value at
    least $1$.  Further, if $x$ is a feasible solution to \ref{pr:main} with
    value $(1-\eps)$ there exists a solution $y$ such that $y \in
    \ST(G)$ and $Ay \le \frac{1}{1-\eps} b$.
  \end{observation}

  \ref{pr:main} is a pure packing LP, albeit in \emph{implicit}
  form. We approximately solve the LP using MWU techniques and in
  particular we use the randomized MWU framework from \cite{cq-18}
  for positive LPs. We review the framework in the next subsection and
  subsequently show how to apply it along with other implementation
  details to achieve the concrete run times that we claim.

\subsection{Randomized MWU Framework for Positive LPs}
The exposition here follows that of~\cite{cq-18}.  Randomized MWU
serves as an LP solver for mixed packing and covering problems. We
only consider pure packing problems in this paper, and therefore we
focus the discussion around randomized MWU on its application to pure
packing problems, which are problems of the form
    \begin{equation}\label{pr:pure-pack}
      \max\{\ip{c}{x} : Ax \le \1, x \ge 0\}
    \end{equation}
    where $A \in \R_+^{k \times n}$ and $c \in \R_+^n$. The pseudocode
    for randomized MWU is given in Figure~\ref{alg:rand-pack}.

    \begin{figure}[t]
      \centering
      \fbox{\parbox{0.85\linewidth}
      {
        $\randpack(A \in \R_+^{k \times n}, c \in \R_+^n, \eps \in (0,1))$

        \begin{algorithmic}
          \STATE $w \gets \1$; $x \gets 0$; $t \gets 0$; $\eta \gets \frac{\ln k}{\eps}$
          \WHILE{$t < 1$}
            \STATE let $y \in \R_{\ge 0}^n$ such that
              $\ip{c}{y} \ge (1 - O(\eps))\max\{\ip{c}{z} : \ip{Az}{w} \le \ip{\1}{w}, z \ge 0\}$
            \STATE $\delta \displaystyle\gets \min\left\{\frac{\eps}{\eta}\cdot
              \min_{i \in [m]}\frac{1}{\ip{e_i}{Ay}}, 1- t \right\}$
            \STATE $x \gets x + \delta y$
            \STATE pick $\theta$ uniformly at random from $[0,1]$
            \STATE $L \gets \{i \in [k] : \theta \le \delta \eta \ip{e_i}{Ay} /\eps\}$
            \FOR{$i \in L$}
              \STATE $w_i \gets w_i \cdot \exp(\eps)$
            \ENDFOR
            \STATE $t \gets t + \delta$
          \ENDWHILE
          \RETURN{$x$}
        \end{algorithmic}
      }}
      \caption{Randomized MWU from~\cite{cq-18} written for pure packing
        problems.}
      \label{alg:rand-pack}
    \end{figure}

    Randomized MWU leverages the simplicity of Lagrangian relaxations
    to speed up the running time. The idea is to collapse the set of
    $k$ constraints to a single, more manageable constraint. In
    particular, a non-negative weight $w_i$ is maintained for each of the
    $k$ constraints, and the relaxation we consider is
    \begin{equation}
      \max\{\ip{c}{z} : \ip{Az}{w} \le \ip{\1}{w}, z \ge 0\}.
    \end{equation}
    We note two important points about this relaxation. The first important point is that there exists an optimal solution
    to the relaxation that is supported on a single $j \in [n]$.
    Randomized MWU takes advantage of this fact to be able to compute
    the optimal solution fast. The second point is that randomized MWU
    only requires an approximation to the relaxation, which can be
    helpful in obtaining faster running times. The LP solver we
    introduce in this paper takes advantage of this last point.

    After computing the solution $y$ to the relaxation, randomized
    MWU adds the solution $y$ scaled by $\delta$ to the solution
    $x$ that it maintains throughout the algorithm. To obtain a fast running time,
    the algorithm utilizes non-uniform step sizes, which were introduced
    in~\cite{gk-07}.

    A main difference between randomized MWU and the deterministic
    variants of MWU is the randomized weight update. The randomized
    weight update essentially computes the deterministic update in
    expectation. The efficiency gains come from the fact that the
    weight updates are all correlated via the random variable
    $\theta$. The total number of iterations is bounded by
    $O(\frac{m \log m}{\eps^2})$ with high probability.

    Note that in the pseudocode in Figure~\ref{alg:rand-pack} the key
    steps that determine the running time are the following in each
    iteration: (i) finding an (approximate) solution $y$ (ii)
    computing the step size $\delta$ and (iii) finding the indices of
    the constraints in $L$ whose weights need to be updated based on
    the choice of $\theta$. These depend on the specific problem and
    implementation. Note that \ref{pr:main} is an implicit LP with an
    exponential number of variables and hence we summarize below
    the key properties of the algorithm.

     \begin{theorem}[\cite{cq-18}]\label{thm:rand-mwu2}
       Let $A \in \R_+^{k \times n}$, $c \in \R_+^n$ which may be
       specified implicitly. Let $x^*$ be an optimal solution to the
       LP~(\ref{pr:pure-pack}). With probability
       $1 - \frac{1}{\poly(k)}$, $\randpack(A, c, \eps)$ returns a
       fractional solution $x$ such that $\ip{c}{x} \ge (1-\eps) \ip{c}{x^*}$
       and $Ax \le (1 + O(\eps))\1$ in $O(k \log k / \eps^2)$
       iterations of the while loop.  Each weight $w_i$ increases
       along integral powers of $\exp(\eps)$ from $1$ to at most
       $\exp(\ln(k) / \eps)$. The sum of the weights
       $\sum_{i=1}^k w_i$ is bounded by $O(\exp(\ln(k) / \eps))$.
    \end{theorem}

  \subsection{Fast implementation of randomized MWU for \cst}
  Our task here is to show how the algorithm can be implemented
  efficiently.  Some of the ideas here have been implicitly present in
  earlier work, however, the precise details have not been spelled out
  and we believe that it is important to specify them. In
  Figure~\ref{alg:rand-mwu-CST}, we provide the pseudocode
  $\crossrand$ that specializes the generic randomized MWU algorithm
  to our specific problem here --- the Lagrangian relaxation oracle is
  to find an MST in $G$ with appropriate edge lengths and we discuss
  it below. Note that in our packing formulation \ref{pr:main} the
  constraint matrix is implicit and is defined via the matrix $A$ of
  \ref{pr:compact}; although this may be a bit confusing we feel that
  it was more natural to formulate \ref{pr:compact} via the use of the
  packing matrix $A$. The correctness of the algorithm and its output
  follow from Theorem~\ref{thm:rand-mwu2} and our main goal is to show
  that it can implemented in $\tilde{O}(N/\eps^2)$ time where $N$ is
  the number of non-zeroes in the matrix $A$.

  Recall that the algorithm maintains a weight $w_i$ for each of the
  $k$ non-trivial packing constraints; the weights are initialized to
  $1$ and monotonically go up. In each step the algorithm solves a
  Lagrangian relaxation which corresponds to collapsing the $k$
  constraints into a single constraint by taking the weighted linear
  combination of the constraints.

   $$\max \sum_T x_T \text{~subject to~} \sum_T x_T \sum_{i=1}^k w_i
  (A\1_T)_i  \le \sum_i w_i b_i \text{~and~} x \ge 0.$$

  It is easy to see that there exists an optimal solution to the
  Lagrangian relaxation that is supported on a single tree. The
  solution is of the form
  \[
    \frac{\ip{w}{\1}}{\sum_{i=1}^k w_i (A\1_T)_i} \cdot e_T
    \quad\text{ where }\quad
    T = \argmin_{T \in \cT} \sum_{i=1}^k w_i (A\1_T)_i.
  \]
  We can rewrite the sum $\sum_{i=1}^k w_i (A\1_T)_i$ as
  $\sum_{e \in T} \sum_{i=1}^k w_iA_{i,e}$, implying that $T$ is a
  minimum spanning tree of $G$ where the edge length of $e$ is
  $\sum_{i=1}^k w_iA_{i,e}$.

  Following Theorem~\ref{thm:rand-mwu2} we see that the overall time
  to update the weights is $O(k \log k/\eps^2)$ since each of the $k$
  weights is only updated $O(\log k/\eps^2)$ times.  The overall
  running time of the algorithm depends on three key aspects: (i)
  total time to maintain an approximate solution to the Lagrangian
  relaxation at each iteration, (ii) total time to find the list of
  weights to update in each iteration after the random choice of
  $\theta$, and (iii) total time to find the step size $\delta$ in each
  iteration. We now describe a specific implementation that leads to
  an $\tO(N/\eps^2)$ bound.

  Given weights $w_i$ for $i \in [k]$, the Lagrangian relaxation
  requires us to compute an MST according to \emph{lengths} on the
  edges that are derived from the weights on constraints.  The length
  of edge $e$, with respect to the weight vector $w$, which we denote
  by $\mstlength_w(e)$, is defined as $\sum_{i=1}^k w_i A_{i,e}$. For
  the most part we omit $w$ and use $\mstlength(e)$ to avoid
  notational overload. We maintain the edge lengths for each
  $e \in E$.  It would require too much time to compute an MST in each
  iteration from scratch. We therefore use the dynamic MST of Holm et
  al.\ \cite{hlt-01} to maintain a minimum spanning tree subject to
  edge insertions and deletions where both operations take
  $O(\log^4n)$ amortized time. When a weight $w_i$ corresponding to
  the packing constraint $i$ increases, we need to update
  $\mstlength(e)$ for all the edges $e$ in the set
  $\{e' \in E : A_{i,e'} > 0\}$. Each update to an edge length
  requires an update to the data structure that maintains the
  MST. From Theorem~\ref{thm:rand-mwu2}, the total number of times a
  weight $w_i$ changes is $O(\frac{\log k}{\eps^2})$.  Each time $w_i$
  changes we update $\rho(e)$ for all $e$ such that $A_{i,e} > 0$.
  This implies then that the total number of updates to the edge
  lengths is $O(N \log k /\eps^2)$ and hence the total time to
  maintain $T$ is $O(\frac{N \log k \log^4 n}{\eps^2})$.  Although we
  can maintain an exact MST as above in the desired run time we note
  that the approximation guarantee of the overall algorithm will hold
  true even if we maintained a $(1+\eps)$-approximate MST. We need to
  use this flexibility in order to avoid the bottleneck of updating
  weights in each iteration.

  Now we address the second issue that affects the run time.  In each
  iteration we need to efficiently find the list $L$ of constraints
  whose weights need to be increased.  $L \subseteq [k]$ is the set of
  indices $i$ satisfying
  $\theta \le \delta\eta(A\1_{T})_i / (b_i \eps)$ where $T$ is the
  spanning tree that was used in the current iteration as the solution
  to the Lagrangian relaxation. The difficulty in computing $L$ is in
  maintaining the values $\{(A\1_{T})_i / b_i : i \in [k]\}$. For a
  given constraint $i$ and tree $T$ we let $\constcost_T(i)$ denote
  $(A\1_{T})_i / b_i = \frac{1}{b_i}\sum_{e \in T} A_{i,e}$; we omit
  $T$ in the notation when it is clear. We maintain the $\constcost(i)$
  values in a standard balanced binary search tree $S$: given
  $\theta$, $\delta$, $\eta$, and $\eps$ we can use $S$ to report $L$
  in time $O(|L| + \log k)$. Note that each time we retrieve an
  $i \in L$ we increase $w_i$ by a $(1+\eps)$ factor and the total
  number of weight updates to $i$ is $O(\log k/\eps^2)$. Thus we can
  charge the retrieval time from $S$ to the total number of weight
  updates which is $O(k \log k/\eps^2)$. The main difficulty is in
  updating the values of $\constcost_T(i)$ in $S$ each time the tree
  $T$ changes.  Recall that $T$ changes a total of
  $O(N \log k/\eps^2)$ times via single edge updates. Suppose
  $\mstlength(e)$ increases and this requires an update to $T$; $e$ is
  removed and a new edge $e'$ enters into the tree. We now need to
  update in $S$ the value $\constcost(i)$ if $A_{i,e'} > 0$ or
  $A_{i,e} > 0$; updating a single value in $S$ takes $O(\log k)$
  time. The number of constraints affected is $N_e+N_{e'}$ where $N_e$
  and $N_{e'}$ are the number of non-zeroes in the columns of $A$
  corresponding to $e$ and $e'$ respectively.  The total number of
  edge updates to the MST data structure, as we outlined, is
  $O(N \log k/\eps^2)$. Hence, if we update the values in $S$ naively,
  then the total number of updates to $S$ would be
  $O(\Delta_0 N \log k/\eps^2)$ where $\Delta_0$ is the maximum number
  of non-zeroes in any column of $A$. In the case of \bdmst we have
  $\Delta_0 = 2$ but in the setting of \cst we may have
  $\Delta_0 = k$. Thus, using the dynamic MST data structure and a
  straightforward use of $S$ gives a total running time of
  $\tO(\Delta_0 N/\eps^2)$. Note that for \bdmst $\Delta_0 = 2$ and
  $N=2m$, so this already leads to an $\tO(m/\eps^2)$ running time.
  However we seek to obtain near-linear time even when $\Delta_0$ can
  be large.

  To avoid having to update $S$ so frequently, we wish to limit the
  number of times that $\mstlength(e)$ is changed. Here is where we
  can take advantage of the fact that a $(1+\eps)$-approximation to
  the MST suffices. In the basic implementation we updated the tree
  $T$ whenever $\mstlength(e)$ increased for some $e \in T$. In a
  modified implementation we update $T$ only when $\mstlength(e)$
  increases by a $(1+\eps)$-factor. For this purpose we maintain
  $\mstlength(e)$ as well as another quantity
  $\approxmstlength(e)$. We maintain the invariant that
  $\approxmstlength(e) \le \mstlength(e) \le (1+\eps)
  \approxmstlength(e)$.  The dynamic MST data structure is based on
  the approximate edge lengths $\approxmstlength(e)$, $e \in E$ and
  the invariant guarantees that at any point and for any tree $T$,
  $\mstlength(T) \le (1+\eps)\approxmstlength(T)$. Thus if $T$ is an
  MST with respect to the approximate lengths then it is a
  $(1+\eps)$-approximation to the MST with respect to the actual
  lengths. Note that we maintain the exact $\mstlength$ values and
  whenever $\mstlength(e)$ is updated the algorithm checks
  $\approxmstlength(e)$ and updates it to maintain the invariant (and
  updating $\approxmstlength(e)$ may trigger an udpate to the
  tree). It is not hard to see that this approximate length
  maintenance does not affect the total time to maintain the edge
  lengths and the total time to update the (approximate) MST. How many
  times does $\approxmstlength(e)$ get updated?  The initial value of
  $\mstlength(e) = \sum_{i=1}^k A_{i,e}$ since $w_i = 1$ at the start
  of the algorithm. The final value is $\sum_{i=1}^k A_{i,e} w_i$
  where the $w_i$ are the weights at the end of the algorithm.
  Theorem~\ref{thm:rand-mwu2} guarantees we have
  $\sum_i w_i \le \exp(\ln k/\eps)$. This implies that $\mstlength(e)$
  increases by a $(1+\eps)$-factor at most $O(\ln k/\eps^2)$ times.
  Therefore $\approxmstlength(e)$ changes only $O(\ln k/\eps^2)$
  times.

  Now we see how maintaining an approximate solution to the MST
  helps us improve the time to update $\constcost$
  values in $S$. Recall that we need to update $S$ when $T$ changes. $T$
  changes when $\approxmstlength(e)$ increases for some $e \in T$ and
  is potentially replaced by an edge $e'$. As we described earlier,
  each such change requires $N_e + N_{e'}$ updates to $S$. We can
  charge the $N_{e'}$ updates to the time that $e'$
  leaves the tree $T$ (if it never leaves then we can charge it once
  to $e'$). Since each edge $e$ leaves the tree $O(\log k/\eps^2)$
  times, the total number of updates is
  $O((\sum_e N_e) \log k/\eps^2)$ which is $O(N \log k/\eps^2)$. Each
  update in the binary search tree takes $O(\log k)$ time and thus the
  total time for finding $L$ in each iteration via the binary search
  tree $S$ is $O(N \log^2 k/\eps^2)$.

  Finally  we describe how to find the step size $\delta$ in each
  iteration. To find $\delta$ we need to compute
  $\min_{i \in [k]}\frac{b_i}{(A\1_{T})_i} = \min_{i \in [k]}
  \frac{1}{\gamma_T(i)}$. Since we maintain $\gamma_T(i)$ values in
  a binary search tree $S$ this is easy to find in $\tO(1)$ time per
  iteration and the total number of iterations is $O(k \log k/\eps^2)$.

   \begin{figure}[t]
    \centering
    \fbox{\parbox{0.9\linewidth}
    {
      $\crossrand(G = (V,E), A \in [0,1]^{k \times m}, b \in \R_+^k, \eps \in (0,1))$

      \begin{algorithmic}
        \STATE $w \gets \1$; $x \gets 0$; $t \gets 0$; $\eta \gets
        (\ln m) / \eps$
        \STATE For each $e$ define $\mstlength_e :=  \sum_{i=1}^kA_{i,e} w_i$
        \WHILE{$t < 1$}
          \STATE $T$ is $(1+O(\eps))$-approximate MST in $G$ with respect to edge lengths
          $\mstlength(e), e \in E$
          \STATE $y \gets \displaystyle\frac{\ip{w}{\1}}
            {\sum_{i = 1}^k w_i\cdot (A\1_{T})_i}\cdot e_{T}$
          \STATE $\delta \displaystyle\leftarrow \min\left\{\frac{\eps}{\eta}\cdot
            \min_{i \in [k]}\frac{b_i}{(A\1_{T})_i}, 1- t \right\}$
          \STATE $x \gets x + \delta y$
          \STATE pick $\theta$ uniformly at random from $[0,1]$
          \STATE $L \gets \{i \in [k] \mid (\theta \eps)/(\delta\eta)
          \le (A\1_{T})_i/b_i\}$
          \FOR{$i \in L$}
            \STATE $w_i \gets w_i \cdot \exp(\eps)$
          \ENDFOR
          \STATE $t \gets t + \delta$
        \ENDWHILE
        \RETURN{$x$}
      \end{algorithmic}
    }}
    \caption{Specialization of generic randomized MWU algorithm to
      solve \ref{pr:main}.}
    \label{alg:rand-mwu-CST}
  \end{figure}

  \begin{figure}[t]
    \centering
    \fbox{\parbox{0.9\linewidth}
    {
      $\fastcrossrand(G = (V,E), A \in [0,1]^{k \times m}, b \in \R_+^k, \eps \in (0,1))$

      \begin{algorithmic}
        \STATE $w \gets \1$; $x \gets 0$; $t \gets 0$; $\eta \gets
        (\ln m) / \eps$
        \STATE For each $e \in E$, $\mstlength_e, \approxmstlength_e
        \gets \sum_{i = 1}^k A_{i,e}w_i$
        \STATE $T \gets$ dynamic MST in $G$ according
        to $\approxmstlength_e$ lengths
        \STATE $S \gets$ balanced binary search tree on $[k]$ where
          index $i \in [k]$ has value $(A\1_{T})_i / b_i$
        \WHILE{$t < 1$}
          \STATE $y \gets \displaystyle\frac{\ip{w}{\1}}
            {\sum_{i = 1}^k w_i\cdot (A\1_{T})_i}\cdot e_{T}$
          \STATE $\delta \displaystyle\leftarrow \min\left\{\frac{\eps}{\eta}\cdot
            \min_{i \in [k]}\frac{b_i}{(A\1_{T})_i}, 1- t \right\}$
          \STATE $x \gets x + \delta y$
          \STATE pick $\theta$ uniformly at random from $[0,1]$
          \STATE $L \gets$ all indices in $S$ with value at least
            $(\theta \eps)/(\delta\eta)$
          \FOR{$i \in L$}
            \STATE $w_i \gets w_i \cdot \exp(\eps)$
            \FOR{$e \in E$ such that $A_{i,e} > 0$}
              \STATE $\mstlength_e \gets \mstlength_e + (A_{i,e}w_i)(1 - e^{-\eps})$
              \IF{$\mstlength_e \ge (1 + \eps)\bar{\mstlength}_e$}
                \STATE $\bar{\mstlength}_e \gets \mstlength_e$
                \STATE update $T$ with new cost $\bar{\mstlength}_e$ on edge $e$
                \STATE if $T$ changes, update $S$
              \ENDIF
            \ENDFOR
          \ENDFOR
          \STATE $t \gets t + \delta$
        \ENDWHILE
        \RETURN{$x$}
      \end{algorithmic}
    }}
    \caption{Fast implementation of randomized MWU algorithm for
      solving \ref{pr:main}.}
    \label{alg:fast-rand-mwu-CST}
  \end{figure}

  Figure~\ref{alg:fast-rand-mwu-CST} describes the fast implementation
  that we described above.
  The following theorem proves Theorem~\ref{thm:intro-lp-solve} but only
  for the feasibility LP of \cst. Combining the following theorem with
  the observations in the following section that allow one to incorporate 
  cost as a constraint in the feasibility LP~\ref{pr:main} proves 
  Theorem~\ref{thm:intro-lp-solve}.
  \begin{theorem}\label{lem:cross-solve}
    Let $G = (V,E)$ be a graph with $n$ vertices and $m$ edges,
    $A \in [0,1]^{k \times m}$, $b \in \R_+^k$, and $\eps \in (0,1)$
    and assume there exists a feasible solution
    to~\ref{pr:compact}. With probability at least
    $1 - \frac{1}{\poly(k)}$, $\fastcrossrand(G, A, b, \eps)$ runs in
    $O(\frac{N (\log k \log^4 n + \log^2k)}{\eps^2})$ time where $N$ is
    the number of nonzero coefficients in $A$ and returns an implicit solution $x$ such that
    $Ax \le (1 + O(\eps))b$ and $\sum_T x_T \ge (1-O(\eps))$. By scaling, 
    $x$ can be converted into a convex combination $x'$ of spanning trees such that 
    $Ax' \le (1+O(\eps))b$. If $\sum_T x_T < (1-O(\eps))$ then \ref{pr:compact} is not feasible.
  \end{theorem}

  \begin{remark}
    Note that the convex combination of spanning trees returned by the
    algorithm is implicit since we update $T$ dynamically via the MST
    data structure. The total representation size of this implicit
    representation is upper bounded by $O(k \log k/\eps^2)$ since
    this is an upper bound on the total number of updates to the tree.
  \end{remark}

\subsection{Solving the minimum cost LP}\label{sec:costs}
We have seen how to solve the approximate feasibility version of
\ref{pr:main} and hence the approximate feasibility version of
\ref{pr:compact} in $\tO(N/\eps^2)$ time. Now we describe how we can
also solve the minimum cost version of \ref{pr:compact}. We only
handle non-negative costs.  Suppose we wish to check
whether there exists a feasible solution to \ref{pr:compact} with cost at most a given
value $B$. Then we can add the constraint $\sum_{e} c_e x_e \le B$
to the set of packing constraints and reduce to the
feasibility problem that we already know how to solve.  To find a
feasible solution with smallest cost, say $\OPT$, we can do binary
search. If we knew that $\OPT \in [L, U]$ where $L$ is a lower bound
and $U$ is an upper bound then standard binary search would allow us
to obtain a $(1+\eps)$-approximation to the cost using
$O(\log \frac{U}{L \eps })$ calls to the feasibility problem.

We will assume that costs are non-negative integers. This is without
loss of generality if they are assumed to be rational. It may be the
case that $\OPT = 0$. This can be checked by solving for
feasibility in the subgraph of $G$ restricted to zero-cost
edges. Otherwise $L \ge 1$ and $U \le C$ where $C = \sum_e c_e$ and
binary search will require $O(\log \frac{C}{\eps})$ calls to the
feasibility algorithm. This bound is reasonable but not
strongly polynomial in the input size. We now argue how to find $L$ and $U$ such that
$U/L \le n$ and $\OPT \in [L, U]$.  We sort edges
by cost and renumbering edges we assume that
$c_1 \le c_2 \le \ldots \le c_m$.  We do binary search to find the
smallest $i$ such that there is an (approximate) feasible solution in
the subgraph of $G$ induced by the first $i$ edges in the ordering.
This takes $O(\log m)$ calls to the feasibility oracle. If
$c_i = 0$ then we stop since we found a zero-cost feasible solution.
Otherwise $\OPT \ge c_i$ and moreover we see that $\OPT \le (n-1)c_i$
since we have a feasible solution using only edges $1$ to $i$, all of
whose costs are at most $c_i$. Thus $U \le (n-1)c_i$. We have
therefore found $L, U$ such that $U/L \le n$ and $\OPT \in [L,U]$.
Now we do binary search in $[L,U]$ to find a $(1+\eps)$-approximation to the
optimum cost. Thus, the total number of calls to the feasibility oracle 
is $O(\log m + \log \frac{n}{\eps})$.

%% file: putting-things-together.tex
\section{Putting things together and extensions}
\label{sec:putting-together}
The preceding sections proved the main technical results corresponding
to Theorems \ref{thm:intro-lp-solve}, \ref{thm:intro-sparsification}, and
\ref{thm:intro-swap-round}. In this section we give formal proofs
of the corollaries stated in Section~\ref{sec:intro}.
In addition, we also describe some extensions and other related
results that follow from the ideas in the preceding sections.

\subsection{Proofs of corollaries}
We start with Corollary~\ref{cor:intro-bdmst}.
\begin{proof}[Proof of Corollary~\ref{cor:intro-bdmst}.]
  Let $G = (V,E)$ be the input graph with $m$ edges and $n$ vertices
  and let $\eps > 0$. Consider the LP relaxation to test whether $G$
  has a spanning tree with degree at most a given parameter $B'$.
  Theorem~\ref{thm:intro-lp-solve} implies that there exists a
  randomized algorithm that, with high probability, either correctly
  determines that there is no feasible solution to the LP, or outputs
  a fractional spanning tree $y \in \ST(G)$ such that
  $\sum_{e \in \delta(v)} y_e \le (1 + \eps)B'$ for all $v$.  Using
  the algorithm, we can do binary search over the integers from $2$ to
  $n-1$ to find the smallest value $B$ for which the algorithm outputs
  a solution.  We will focus on the scenario where the algorithm from
  Theorem~\ref{thm:intro-lp-solve} is correct in each of the
  $O(\log n)$ calls in the binary search procedure; this happens with
  high probability. For the value of $B$ found in the binary search,
  let $y \in \ST(G)$ be the solution that is output; we have
  $\sum_{e \in \delta(v)} y_e \le (1 + \eps)B$ for all $v$.  Since the
  approximate LP solver correctly  reported infeasibility for all $B' < B$, we
  have $B - 1 < B^*$, which implies $B \le B^*$. As there is a
  feasible fractional spanning tree $y$ such that
  $\sum_{e \in \delta(v)} y_e \le \ceil{(1+\eps)B}$ for all $v$, the
  result of~\cite{ls-15} implies that $B^* \le \ceil{(1+\eps)B} + 1$.

  Regarding the run time, each call to the algorithm in
  Theorem~\ref{thm:intro-lp-solve} takes $\tO(m/\eps^2)$ time since
  $N = O(m)$ in the setting of \bdst.  Binary search adds only an
  $O(\log n)$ multiplicative-factor overhead, leading to the claimed
  running time.
\end{proof}

We next prove Corollary~\ref{cor:intro-sparse-bdmst}. The algorithm
described in the corollary takes advantage of
Theorem~\ref{thm:intro-sparsification} to sparsify the input graph,
then runs the \Furer-Raghavachari algorithm \cite{FurerR94}.

\begin{proof}[Proof of Corollary~\ref{cor:intro-sparse-bdmst}.]
  We will assume that the input is a graph $G = (V,E)$ $m$ edges and
  $n$ vertices. Let $\eps > 0$. As in the proof of
  Corollary~\ref{cor:intro-bdmst}, we can use
  Theorem~\ref{thm:intro-lp-solve} and binary search to find in
  $\tO(\frac{m}{\eps^2})$ time, with high probability, a fractional
  spanning tree $x \in \ST(G)$ such that
  $\sum_{e \in \delta(v)} x_e \le (1+\eps)B^*$ for all $v \in V$. By
  Theorem~\ref{thm:intro-sparsification}, we can use random sampling
  based on $x$ to obtain a subgraph $G' = (V,E')$ of $G$ such that
  with high probability we have
  $\abs{E'} = O(\frac{n \log (n + m)} {\eps^2})$ and there exists a
  fractional solution $z \in \ST(G)$ in the support of $G'$ such that
  $\sum_{e \in \delta(v)} z_e \le (1+3\eps)(1+\eps)B^* \le (1+7\eps)B^*$ for all
  $v \in V$ (for sufficiently small $\eps$). The result of~\cite{ls-15} implies there
  exists a spanning tree $T$ in $G'$ such that
  $\max_{v \in V} \deg_T(v) \le \ceil{(1+7\eps)B^*}+1$. For a graph $H$
  on $n$ vertices and $m$ edges, the algorithm of~\cite{FurerR94} runs
  in $\tO(mn)$ time and outputs a spanning tree $T$ such that the
  maximum degree of $T$ is at most $\OPT(H)  + 1$, where
  $\OPT(H) = \min_{T \in \cT(H)}\max_{v \in V} \deg_T(v)$. Thus, the
  algorithm of~\cite{FurerR94} when applied to $G'$, outputs a
  spanning tree with degree at most $\ceil{(1+7\eps)B^*} +2$, and runs
  in $\tO(\frac{n^2}{\eps^2})$ time.
  \end{proof}

  \medskip
  \noindent
  {\bf Non-uniform degree bounds:} We briefly sketch some details
  regarding Remark~\ref{remark:intro-non-uniform}. First, we note that
  the algorithm for solving the LP relaxation handles the non-uniform
  degree bound case in $\tilde{O}(m/\eps^2)$ time. It either certifies
  that the given bounds are infeasible or outputs a fractional
  solution with degree at most $(1+\eps)B_v$ for each $v$.  We can
  then apply the result in \cite{ls-15} to know that there exists a
  spanning tree $T$ in which the degree of each $v$ is at most
  $\ceil{(1+\eps)B_v}+1$. We can apply sparsification from
  Theorem~\ref{thm:intro-sparsification} to the fractional solution to
  obtain a sparse subgraph that contains a near-optimal fractional
  solution. It remains to observe that the \Furer-Raghavachari
  algorithm can be used even in the non-uniform setting via a simple
  reduction to the uniform setting. This was noted in prior work
  \cite{kr-03,ls-15} and we provide the details in
  Section~\ref{sec:fr-non-unif}.  This results in an
  $\tilde{O}(n^2/\eps^2)$ time algorithm that either decides that the
  given non-uniform degree bounds are infeasible or outputs a spanning
  tree in which the degree of each node $v$ is at most
  $\ceil{(1+\eps)B_v}+2$.

  \medskip

We finally prove Corollary~\ref{cor:intro-solve-round}. This corollary
combines the LP solver of Theorem~\ref{thm:intro-lp-solve} and the
fast implementation of randomized swap rounding of
Theorem~\ref{thm:intro-swap-round} to give a fast algorithm for \cmst.
\begin{proof}[Proof of Corollary~\ref{cor:intro-solve-round}.]
  Let $G = (V,E)$ be the input graph and $\eps > 0$. We want to solve
  $\min\{c^Tx : Ax \le b, x \in \ST(G)\}$.  By
  Theorem~\ref{thm:intro-lp-solve}, there exists a randomized
  algorithm that runs in $\tO(N / \eps^2)$ time and with high
  probability certifies the LP is infeasible or outputs $y \in \ST(G)$
  such that $c^Ty \le (1+\eps)\OPT$ and $Ay \le (1+\eps)b$. We then
  apply the fast implementation of randomized swap rounding of
  Theorem~\ref{thm:intro-swap-round} to $y$, which runs in
  $\tO(m/\eps^2)$ time and outputs a spanning tree $T$. If we only
  considered the feasibility version (i.e.\ find $x \in \ST(G)$ such
  that $Ax \le b)$, then the existing results on swap
  rounding~\cite{cvz-10} imply that
  $A\1_T \le \min\{O(\log k / \log \log k)b, (1 + \eps)b + O(\log
  k)/\eps^2\}$ with high probability. In the cost version,
  \cite{cvz-10} implies that
  $A\1_T \le \min\{O(\log k / \log \log k)b, (1 + \eps)b + O(\log
  k)/\eps^2\}$ with high probability, and
  $\E[c^T\1_T] \le \sum_e c^T x$. Thus, the cost is preserved only in
  expectation. We can, however, apply Markov's inequality and conclude
  that
  $\Pr[c^T\1_T \ge (1+\eps)c^Tx] \le \frac{1}{1+\eps} \le 1 -
  \frac{\eps}{2}$ (for $\eps$ sufficiently small). For a suitable choice of the
  high probability bound, we have that
  $A\1_T \le \min\{O(\log k / \log \log k)b, (1 + \eps)b + O(\log
  k)/\eps^2\}$ \emph{and} $c^T\1_T \le (1+\eps)c^Tx$ with probability
  at least $\frac{\eps}{2} - \frac{1}{2n^2}$.  We can assume that
  $\eps > \frac{1}{n^2}$, for otherwise the $\frac{1}{\eps^2}$ dependence in the
  run time of the approximate LP algorithm is not meaningful; one can use other
  techniques including an exact LP solver. Thus, with probability at
  least $\frac{\eps}{4}$, we have
  $c^T\1_T \le (1+O(\eps))c^Tx \le (1+O(\eps))\OPT$ and
  $A\1_T \le \min\{O(\log k / \log \log k)b, (1 + \eps)b + O(\log
  k)/\eps^2\}$.  To boost the $\frac{\eps}{4}$ probability of success, we can
  repeat the rounding algorithm $O(\frac{\log n}{\eps})$ times
  independently; with high probability, one of the trees will yield
  the desired bicriteria approximation.
\end{proof}

\subsection{A combinatorial algorithm for the non-uniform degree bound case}
\label{sec:fr-non-unif}
For a graph on $m$ edges and $n$ vertices, we show that the \Furer-Raghavachari 
algorithm can be used to solve the
non-uniform degree bound case in $\tO(mn)$ time. We do this by showing that
the reduction in~\cite{kr-03} from the non-uniform degree bound case to the 
uniform case can be done implicitly. We provide the details here as we believe they
have not been specified elsewhere.

Given a graph $G = (V,E)$ on $n = \abs{V}$ vertices and $m = \abs{E}$
edges and degree bounds $B_v$ for each $v \in V$, the reduction
in~\cite{kr-03} constructs a new graph by copying $G$ and for every $v
\in V$, adds $n - B_v$ auxiliary nodes of degree $1$ that are
connected only to $v$.  We refer to this new graph as $R(G)$.  It is
an easy observation that there is a spanning tree $T$ in $G$
satisfying the given degree bounds if and only if there is a spanning
tree in $R(G)$ with degree at most $n$. Moreover, if $T'$ is a
spanning tree in $R(G)$ with maximum degree $n+c$, then restricting
$T'$ to $V(G)$ yields a spanning tree $T$ of $G$ in which the degree
of each $v$ is at most $B_v + c$.  We can therefore run the
\Furer-Raghavachari algorithm on $R(G)$ with uniform degree bound
$n$. If the algorithm reports no such tree exists, then there is no
spanning tree of $G$ satisfying the given non-uniform degree
bounds. If a spanning tree $T'$ of $R(G)$ is returned,
then we obtain a spanning tree $T$ in $G$ such that
$\deg_T(v) \le B_v + 1$ for all $v \in V$. The main issue to resolve is
the running time of the algorithm; the graph $R(G)$
could have $\Omega(n^2)$ vertices and
$\Omega(n^2)$ edges.

We describe how one can run the \Furer-Raghavachari algorithm on the
graph $R(G)$ \emph{implicitly} so that it runs in $\tO(mn)$ time.  We
start by recalling the \Furer-Raghavachari algorithm. Let $H$ be the
input graph with $m_H$ edges and $n_H$ vertices.  The algorithm relies
on a subroutine we refer to as $\reduce$. $\reduce$ takes a spanning
tree $T$ of $H$ as input, runs in $\tilde{O}(m_H)$ time
and outputs a bit $b \in\{0,1\}$. If $b$ is
$1$, then $\reduce(T)$ produces a witness proving that the maximum
degree of $T$ is at most $\OPT(G) + 1$ where $\OPT(G)$ is the
minimum maximum degree of a spanning tree of $G$. If $b$ is $0$,
$\reduce$ outputs a new spanning tree $T'$
with the following properties: (i) the maximum degree of $T'$ is
no larger than in $T$ and (ii) the
number vertices in $T'$ with degree equal to the maximum degree of $T$
is at least $1$ less than in $T$. This is accomplished by
a series of edge swaps in $H$. With this
subroutine in hand, the \Furer-Raghavachari algorithm starts with an
arbitrary spanning tree $T$ and repeatedly calls $\reduce(T)$
until it returns $1$ and terminates. The correctness follows immediately from the
correctness of $\reduce$. As for the running time, we observe that the
number of vertices of degree $k$ in $T$ is at most
$\frac{2n_H}{k}$. Furthermore, the algorithm must stop
when the maximum degree is $2$ (unless the graph is an edge which is trivial).
Therefore, the total number of calls to
$\reduce$ is at most $\sum_{k=2}^{n_H-1} \frac{2n_H}{k} = O(n_H\log n_H)$.
This leads to a total running time of $\tO(m_Hn_H)$.

We run \Furer-Raghavachari on $R(G)$ implicitly in 
$\tO(mn)$ time as follows. We show that one can implement $\reduce$ in $\tO(m)$
time in $R(G)$, and that the number of calls to $\reduce$ is at most
$\tO(n)$. Let $T'$ be a spanning tree of $R(G)$. Let $k$ be the maximum
degree of $T'$. Note that we can assume $k \ge 3$ for otherwise we
are done. $\reduce(T')$ begins by finding $D_k$ and
$D_{k-1}$, the sets of vertices of degree $k$ and degree $k-1$ in
$T'$, respectively. As we said, it executes a series of edge swaps of non-tree
edges to find another spanning tree in which $|D_k|$ decreases.
We observe that all the auxiliary vertices and edges incident
to them are present in every spanning tree of $R(G)$.
Thus, they do not play any role in the algorithm $\reduce$ other than
in determining $D_k$ and $D_{k-1}$. It is easy to compute
$\deg_{T'}(v)$ for each $v \in V(R(G))$ in $\tilde{O}(m)$ time at
the start of $\reduce$ and then effectively run the algorithm in $G$.
Thus, we can implement $\reduce(T')$ in $\tO(m)$
time.

All that remains is to show that $\reduce$ is called $\tO(n)$ times.

\begin{lemma}
  Let $G = (V,E)$ be an undirected, connected graph where $n = \abs{V}$ and
  $m = \abs{E}$. For $v \in V$, let $B_v$ be the degree bound for $v$. 
  Let $R(G)$ be the graph $G$ where
  for each $v \in V$, $n - B_v$ auxiliary nodes of degree $1$ are added incident to $v$.
  On input $R(G)$, the \Furer-Raghavachari algorithm calls $\reduce$ $\tO(n)$ times.
\end{lemma}
\begin{proof}
  Given an instance of the non-uniform problem, for ease of analysis,
  we assume that there is a node with degree bound $1$. This is
  easy to ensure by adding a new node $u$ and connecting it by an edge to some
  arbitrary node $v \in V(G)$, and setting $B_u = 1$ and adding $1$ to $B_v$.
  With this technical assumption in place we consider
  the graph $R(G)$. Let $B^* = \min_{T
    \in \cT(R(G))} \max_{v \in V(R(G))} \deg_T(v)$ be the optimal
  degree bound for the graph $R(G)$.  We have $B^* \ge n$ since
  $\min_{v \in V(G)} B_v = 1$. We also have that the maxmimum
  degree of any spanning tree in $R(G)$ is at most $(n-1) + (n-1) =
  2n-2$, where the first $n-1$ comes from the auxiliary nodes and the
  second $n-1$ comes from the degree of the spanning tree on the edges
  in $G$.
      
  The \Furer-Raghavachari algorithm starts with an arbitrary spanning
  tree $T'$ of $R(G)$ and iteratively reduces the number of maximum
  degree vertices by $1$ by calling $\reduce$.  As observed above, the
  maximum degree of a vertex in $T'$ is at most $2n-2$ and the
  algorithm will stop when the maximum degree of $T'$ is $n$ as $B^*
  \ge n$ (note that it might stop earlier).  Recall that $\reduce$
  reduces the number of maximum degree vertices in $T'$ by $1$ or
  produces a witness implying $T'$ is near-optimal. Suppose the
  maximum degree of $T'$ is $(n-1) + k$ for some $k \in [n-1]$.  A
  vertex of degree $(n-1)+k$ in $T'$ has degree at least $k$ in $G$.
  In any spanning tree of $G$, there are at most $\frac{2n}{k}$
  vertices of degree at least $k$.  Thus, the number of times
  $\reduce$ is called is at most $\sum_{k=1}^{n-1} (2n/k) = O(n\log
  n)$.

  The assumption that $B_{\min} = \min_{v \in V(G)} B_v = 1$ can be discharged in the
  preceding analysis by considering the degree of $T'$ in $R(G)$ which
  ranges from $n- B_{min}+1$ and $2n-(B_{\min}+1)$.
\end{proof}

\subsection{Extensions and related problems}
We focused mainly on \bdst, \bdmst and \cmst. Here we briefly discuss
related problems that have also been studied in the literature to which
some of the ideas in this paper apply.

\paragraph{Estimation of value for special cases of \cmst:}
As we remarked in Section~\ref{sec:intro}, various special cases of
\cmst have been studied. For some of these special cases one can
obtain a constant factor violation in the constraint
\cite{OlverZ18,LinharesS18}. We highlight one setting. One can view
\bdmst as a special case of \cmst where the matrix $A$ is a
$\{0,1\}$-matrix with at most $2$ non-zeroes per column (since an edge
participates in only two degree constraints); the result in
\cite{ls-15} has been extended in \cite{KLS12} (see also
\cite{BansalKN10}) to show that if $A$ is a $\{0,1\}$-matrix with at
most $\Delta$ non-zeroes per column, then the fractional solution can
be rounded such that the cost is not violated and each constraint is
violated by at most an addtive bound of
$(\Delta-1)$. Theorem~\ref{thm:intro-lp-solve} allows us to obtain a
near-linear time algorithm to approximate the LP.  Combined with the
known rounding results, this gives estimates of the integer optimum
solution in near-linear time. Thus, the bottleneck in obtaining a
solution, in addition to the value, is the rounding step. Finding
faster iterated rounding algorithms is an interesting open problem
even in restricted settings.

\paragraph{Multiple cost vectors:}
In some applications one has multiple different cost vectors on the
edges, and it is advantageous to find a spanning tree that
simultaneously optimizes these costs. Such multi-criteria problems
have been studied in several contexts.  Let $c_1,c_2,\ldots,c_r$ be
$r$ different cost vectors on the edges (which we assume are all
non-negative). In this setting it is typical to assume that we are
given bounds $B_1, B_2,\ldots, B_r$ and the goal is to find a spanning
tree $T \in \cT(G)$ satisfying the packing constraints such that
$c_j(T) \le B_j$ for $j \in [r]$. We can easily incorporate these
multiple cost bounds as packing constraints and solve the resulting LP
relaxation via techniques outlined in Section~\ref{sec:solve}.  Once
we have the LP solution we note that swap-rounding is oblivious to the
objective, and hence preserves each cost in expectation. With
additional standard tricks one can guarantee that the costs can be
preserved to within an $O(\log r)$ factor while
ensuring that the constraints are satisfied to within the same factor
guaranteed in Corollary~\ref{cor:intro-solve-round}.

\paragraph{Lower bounds:} \bdmst has been generalized to the setting where
there can also be lower bounds on the degree constraints of each
vertex.  \cite{ls-15} and \cite{KLS12} showed that the additive degree
bound guarantees for \bdmst can be extended to the setting with lower
bounds in addition to upper bounds. One can also consider such a
generalization in the context of \cmst. The natural LP relaxation for
such a problem with both lower and upper bounds is of the form $\min
\{ c^Tx : Ax \le b, A'x \ge b', x \in \ST(G) \}$ where $A, A' \in
   [0,1]^{k\times m}, b,b' \in [1,\infty)^k, c \in [0,\infty)^m$.
       Here $A$ corresponds to upper bounds (packing constraints) and
       $A'$ corresponds to lower bounds (covering constraints).  This
       mixed packing and covering LP can also be solved approximately
       in near-linear time by generalizing the ideas in
       Section~\ref{sec:solve}.  Sparsification as well as
       swap-rounding can also be applied since they are oblivious to
       the constraints once the LP is solved.
       The guarantees one obtains via swap rounding are based on negative correlation
       and concentration bounds. They behave slightly differently for lower bounds.
       One obtains a tree $T$ such that $A\1_T \le (1+\eps)b + O(\log k)/\eps^2$
       and $A'\1_T \ge (1-\eps)b' - O(\log k)/\eps^2$ with high
       probability.
       As in the other cases, the LP solution proves the existence of
       good integer solutions based on the known rounding results.

%% file: appendix.tex

\section{Skipped proofs}\label{app:proofs}
  In this section, we present two proofs that were skipped in the main body of
  the paper. In our fast implementation of randomized swap rounding, the input is 
  an implicit representation of a convex combination of spanning trees
  $\sum_{i=1}^h \delta_i\1_{T_i}$. We need to be able to compute the intersection
  $\bigcap_{i=1}^h T_i$ via this implicit representation. The following two claims
  allow us to do so. For sets $A,B$, let $A \setdiff B$ denote the symmetric 
  difference of $A$ and $B$. We omit the proofs as they are straightforward.
  \begin{lemma}\label{lem:set-1}
    Let $S_1,\ldots, S_h$ be sets and consider the universe of elements
    to be the union of all the sets. We have
    \[
      \bigcup_{i=1}^{h-1} (S_i \setdiff S_{i+1})
      =
      \bigcup_{i=1}^h S_i \setminus \left(\bigcap_{i=1}^h S_i\right).
    \]
  \end{lemma}

  \begin{lemma}\label{lem:set-2}
    Let $S_1,\ldots, S_h$ be sets and consider the universe of elements
    to be the union of all the sets. We have
    \[
      S_1 \setminus \left(\bigcup_{i=1}^{h-1} (S_i \setdiff S_{i+1})\right)
      =
      \bigcap_{i=1}^h S_i.
    \]
  \end{lemma}
      
  \begin{proof}[Proof of Lemma~\ref{lem:divide-swap}.]
    By Lemma~\ref{lem:set-2}, $B \setminus \bigcup_{i=1}^{h-1} (E_i \cup E_i')
    = \bigcap_{i=1}^h B_i$. Then $B \setminus \bigcap_{i=1}^h B_i
    = B \setminus \left(B \setminus \bigcup_{i=1}^{h-1} (E_i \cup E_i')\right)
    = B \cap \bigcup_{i=1}^{h-1} (E_i \cup E_i')$. Thus, $\hat{B}$ in
    $\divideswap$ is the same as $B$ with the intersection $\bigcap_{i=1}^hB_i$
    removed.

    To see that $\divideswap$ implements $\swap$, it is easy to see
    via Lemma~\ref{lem:cont-matroid-base}
    that for two bases $B,B'$, if one contracts $A \subseteq B \cap B'$  
    in $B$ and $B'$ to obtain $\hat{B}$ and $\hat{B}'$, respectively,
    then $A \cup \merge(\delta, \hat{B},\delta',\hat{B}') \eqdist 
    \merge(\delta,B,\delta',B')$. We can inductively apply this argument to see 
    $\swap(\delta_1,B_1,\ldots,\delta_h,B_h) \eqdist
    \divideswap(B_1,\{(E_i,E_i')\}_{i=1}^{h-1}, \{\delta_i\}_{i=1}^h)$.
  \end{proof}

  \begin{proof}[Proof of Lemma~\ref{lem:shrink-int}.]
    We first prove the guarantees on the output. Lemma~\ref{lem:set-2}
    shows that $T_1 \setminus \bigcup_{i=1}^{h-1}(E_i \cup E_i') =
    \bigcap_{i=1}^h T_i$. It is therefore clear that $\shrinkint$ contracts
    $\bigcap_{i=1}^h T_i$ in $T_1$, so $\hat{T} = \hat{T}_1$.
    As $\hat{E}_i = \hat{T}.\represented(E_i)$ and $\hat{E}_i' = \hat{T}.\represented(E_i')$,
    by Lemma~\ref{lem:diff-cont}, we have by induction that
    $\hat{T}_{i+1} = \hat{T}_i - \hat{E}_i +
    \hat{E}_i'$ for all $i \in [h-1]$. As we construct $\hat{E}_i$ and
    $\hat{E}_i'$ by adding each edge from $E_i$ and $E_i'$, respectively, it
    is clear that $\abs{E_i} = \abs{\hat{E}_i}$ and $\abs{E_i'} =\abs{\hat{E}_i'}$ for all
    $i$. Furthermore, as $\shrinkint$ contracts
    $T_1 \setminus \bigcup_{i=1}^{h-1}(E_i \cup E_i')$ in $T_1$ to obtain
    $\hat{T}_1 = \hat{T}$, we have $\abs{\hat{T}} \le \min\{n_{T_1},\gamma\}$.

    We next prove the bound on the running time.
    By Lemma~\ref{lem:forest}, the running time of $\copyy$ is
    proportional to size of the tree being copied, so the running
    time of $T_1.\copyy()$ is $\tO(n_{T_1})$. We note that we can compute
    $T_1 \setminus \bigcup_{i=s}^{t-1}(E_i \cup E_i')$ easily as we have the time
    to scan $T$ and $\bigcup_{i=s}^{t-1}(E_i \cup E_i')$. Then any data structure
    that supports $\tO(1)$ lookup
    time will suffice, such as a balanced binary search tree.
    Applying Lemma~\ref{lem:cont-run},
    iterating over $T_1 \setminus \bigcup_{i=1}^{h-1}(E_i \cup E_i')$
    and contracting the edges in $\hat{T}$ requires $\tO(n_{T_1} + \gamma)$ time.
    Finally, $\shrinkint$ iterates
    over all sets in $\{(E_i,E_i')\}_{i=1}^{h-1}$ to create
    $\hat{E}_i,\hat{E}_i'$ for all $i \in [h-1]$. As
    one call to $\represented$ takes $O(1)$ time by Lemma~\ref{lem:forest}, the
    final for loop requires $O(\gamma)$ time. This concludes the proof.
  \end{proof}  

    \section{Disjoint-set data structure}
    The goal of disjoint-set data structures is to maintain a family of disjoint sets
    $\cS = \{S_1,\ldots, S_k\}$ defined over a universe of elements $\cU$ such
    that $\bigcup_{i=1}^k S_i = \cU$. The
    data structure needs to support the ability to combine two sets via
    their union and to be able to quickly identify the containing set of an element
    in $\cU$.
    We need an implementation of the disjoint-set data structure to support
    the following four operations, of which the first three are standard. As a matter
    of notation, for an element $u \in \cU$, let $S_u$ be the set containing $u$.
    \begin{itemize}\itemsep0pt
      \item
        $\makeset(u)$: creates the singleton set $\{u\}$
      \item
        $\findset(u)$: returns the representative of the set $S_u$
      \item
        $\union(u,v)$: replaces $S_u$ and $S_v$ with $S_u \cup S_v$ and
        determines a new representative for $S_u \cup S_v$
      \item
        $\changerep(u,v)$: assuming $S_u = S_v$, if the representative of $S_u$
        is not $v$, then make $v$ the representative of $S_u$
    \end{itemize}

    The first three operations are supported by optimal implementations
    of a disjoint-set data structure~\cite{t-75,t-83}. The only operation that is
    not explicitly supported is $\changerep$. We use this operation in
    Section~\ref{sec:swap} so that when we contract an edge $uv$ in a graph
    by identifying $u$ with $v$, we can guarantee that $v$ will be the
    representative of the set $\{u,v\}$.

    To see that $\changerep$ is easy to implement, we consider the
    so-called union-by-rank heuristic as presented in~\cite{t-83}. The 
    implementation represents sets by binary trees, one for each set. The nodes 
    in the tree corresponding to the set $S_u$ each represent an element in $S_u$.
    The root of the tree is arbitrarily made to be the representative of $S_u$. The
    element corresponding to a node has no bearing on the
    analysis of the data structure. Thus,
    to implement $\changerep(u,v)$, we can access the nodes in the
    tree $S_u$ corresponding to $u$ and $v$ and simply swap the fields
    that determines the elements the nodes contain and fix all relevant pointers.

    As $\changerep$ can easily be implemented in $O(1)$ time without changing
    the data structure, the following theorem from~\cite{t-75} still holds.
    \begin{theorem}[\cite{t-75}]\label{thm:disjoint-set}
      For any sequence of $m$ calls to $\makeset$, $\findset$, 
      $\union$, and $\changerep$ where at least $n$ of the calls are to 
      $\makeset$, the total running time is at most $O(m\alpha(m,n))$, where 
      $\alpha(m,n)$ is the inverse Ackermann
      function.
    \end{theorem}
    
\section{Data structure to represent forests}\label{app:forests-ds}
  In this section, we present the forest data structure we need in order to give a fast 
  implementation of randomized swap rounding. 
  We give pseudocode for the implementation in
  Figure~\ref{fig:forest}. The data structure is initialized via the function $\init$, 
  which takes as input the vertices, edges, and unique edge identifiers of the forest. 
  $\init$ initializes an adjacency list $A$, stores a mapping $f$ of edges to their 
  unique edge identifiers, and creates a disjoint-set data structure $R$ where every 
  vertex initially is in its own set.
  As noted in Section~\ref{sec:prelim}, the main operation
  we need to support is contraction of an edge. The operation $\contract$ contracts
  an edge $uv$ in the forest by identifying the vertices $u$ and $v$ as the same
  vertex. This requires choosing $u$ or $v$ to be the new representative (suppose
  we choose $u$ without loss of generality), merging 
  the sets corresponding to $u$ and $v$ in $R$ while making $u$ the representative of the
  set in $R$, and modifying the adjacency list $A$ to reflect
  the changes corresponding to contracting $uv$ and making $u$ the representative.
  After an edge is contracted, the vertex set changes. We need to support the ability
  to obtain the edge identifier of an edge in the contracted forest. The data structure 
  maintains $f$ under edge contractions and returns unique identifiers with the 
  operation $\origedge$. Given an edge $uv$ that was in the contracted forest at
  some point in the past, we also need to support the ability to obtain the edge in
  the vertex set of the current contracted forest. We do this using $R$, which 
  stores all of the vertices that have been contracted together in disjoint sets.
  This operation is supported by the operation $\represented$. Finally, we can copy the
  graph via the operation $\copyy$, which simply enumerates over the vertices, edges,
  and stored unique identifiers of the edges to create a new data structure.
  
  We implement $f$ as a balanced binary search tree that 
  assumes some arbitrary ordering over all possible edges. Also, each list in the 
  adjacency list will be represented by a balanced binary search tree to support fast 
  lookup times in the lists. 
    \begin{figure}[t]
    \begin{multicols}{2}
      \noindent$\init(V,E)$
      \begin{algorithmic}
        \STATE initialize $A$ to be adjacency list of $(V,E)$
        \STATE for all $e \in E$, set $f(e) = id(e)$
        \STATE $R \gets$ empty disjoint-set data structure
        \STATE for all $v \in V$, $R.\makeset(v)$
      \end{algorithmic}
      
      \medskip

      \noindent$\copyy()$
      \begin{algorithmic}
        \STATE for all $u \in A$, add $u$ to $V'$
        \STATE for all $u \in A$ and all $v \in A[u]$, add $uv$ to $E$
        \STATE assume the id's of edges in $E$ are given by $f$
        \RETURN{$\init(V',E')$}
      \end{algorithmic}      

      \medskip

      \noindent$\represented(uv)$
      \begin{algorithmic}
        \STATE $u_r \gets R[u]$; $v_r \gets R[v]$
        \RETURN{$u_rv_r$}
      \end{algorithmic}  
      
      \columnbreak      

      \noindent$\origedge(uv)$
      \begin{algorithmic}
        \RETURN{$f(uv)$}
      \end{algorithmic}    
      
      \medskip      

      \noindent$\contract(uv, z \in \{u,v\})$
      \begin{algorithmic}
        \STATE assume without loss of generality $z = u$
        \FOR{$w \in A[u] \setminus \{v\}$}
          \STATE replace $u$ in $A[w]$ with $v$
          \STATE add $w$ to $A[v]$
          \STATE $f(vw) \gets f(uw)$
          \STATE remove $f(uw)$
        \ENDFOR
        \STATE remove $u$ from $A$
        \STATE $R.\union(u,v)$
        \STATE $R.\changerep(u,v)$
      \end{algorithmic}
    \end{multicols}
    \caption{Implementation of a data structure to represent forests.}
    \label{fig:forest}
    \end{figure}    
    \begin{lemma}\label{lem:forest}
      Let $F = (V,E)$ be a forest and for all $e\in E$, let $id(e)$ be the unique 
      identifier of $e$. The data structure in Figure~\ref{fig:forest}
      can be initialized via $\init$ in
      $\tO(\abs{V})$ time. For $i = 0,1,\ldots, k$, let $F_i = (V_i,E_i)$ be
      the forest after $i$ calls to $\contract$, so $F_0 = F$ and $F_k$ is the current
      state of the forest. We 
      assume $F_i$ is the state of the graph just prior to the $(i+1)$-th call to 
      $\contract$. The data structure supports the following operations.
      \begin{itemize}\itemsep0pt
        \item
          $\origedge(e)$: input is an edge $e \in E_k$.
          Output is the identifier $id(e)$ that was provided when $e$ was
          added to the data structure. Running time is $\tO(1)$.
        \item
          $\represented(e)$: input is two vertices $u, v \in V_i$ for some 
          $i = 0,1,\ldots, k$. Output is the pair $u_rv_r$, where $u_r$ and $v_r$ 
          are the vertices that $u$ and $v$ correspond to in the contracted forest
          $F_k$. Running time is $O(1)$.
        \item
          $\contract(uv, z \in \{u,v\})$: input is two vertices
          $u,v \in V_k$. The operation contracts $uv$ in $E_k$, setting
          the new vertex in $E_{k+1}$ to be $\{u,v\} \setminus \{z\}$. The
          amortized running time is $\tO(\deg_{F_k}(z))$.
        \item
          $\copyy()$: output is a forest data structure with vertices $V_k$ and
          edges $E_k$ along with the stored edge identifiers. Running
          time is $\tO(\abs{V_k})$.
      \end{itemize}
    \end{lemma}  
  \begin{proof}
      We first describe the invariants maintained for $A$, $f$, and $R$
      to prove correctness of each operation. $A$ maintains the
      adjacency list. $f$ stores the set of id's of the current set of edges. 
      The sets in $R$, a disjoint set data structure over $V$, 
      correspond to the vertices that have been contracted 
      together. It is easy to see that $\init$ guarantees that all of the invariants hold.
      
      We note that $\contract$ is the
      only operation that changes anything structurally, so we show how this
      operation preserves the invariants. Suppose without loss of
      generality that $z = u$. In $\contract$, all neighbors of $u$ need to be made
      neighbors of $v$ except for $v$, all of the edges incident to $u$
      need to be removed, and $u$ needs to be removed as a vertex. All of
      these actions are taken in $\contract$, so $A$ is updated correctly. 
      As $u$ and $v$ are combined as one vertex,
      we must union the sets in $R$ and guarantee that $v$ is the representative.
      As for $f$, we simply carry the edge id's over to the edges that were 
      incident to $u$ and are now incident to $v$. Thus, $\contract$ maintains all
      invariants.
      The correctness of the other operations follows from the fact that the
      invariants for each of the structures $A$, $f$, and $R$ hold as stated.
      
      Regarding the running times, as we noted above, $f$ is implemented via
      a balanced binary search tree and so are the lists in $A$. The running time of the 
      calls to $R$ are $\tO(1)$ amortized time by Theorem~\ref{thm:disjoint-set}. The 
      only operation whose running time is not obvious is $\contract$. The method 
      executes $O(\deg_{F_k}(u))$ operations that all take $\tO(1)$ time. 
      We use the above observations about the implementation to see that the 
      remaining operations take $\tO(1)$ amortized time. 
      Thus, $\contract$ runs in $\tO(\deg_{F_k}(u))$ amortized time.
      It is easy to verify the running times of all of the other operations.
      This concludes the proof.
  \end{proof}
  
    Now that we have the implementation of the forest data structure given
    in Lemma~\ref{lem:forest}, we provide a basic fact about contracting
    edges in trees using the data structure. The lemma states
    that the $\represented$ function allows one to essentially contract subtrees
    without having to explicitly contract them via a data structure.

    \begin{lemma}\label{lem:diff-cont}
      Let $T$ be a spanning tree and $A, A'$ be disjoint sets of edges defined on the
      same set of vertices as $T$. Let $T_1 = T$ and $T_2 = T_1 - A + A'$ and
      assume $T_2$ is also a spanning tree.

      Contract $B \subseteq T_1 \cap T_2$ in $T_1$ and $T_2$ to obtain
      $\hat{T}_1$ and $\hat{T}_2$. Let $\tilde{A} = \hat{T}_1.\represented(A)$ and
      $\tilde{A}' = \hat{T}_1.\represented(A')$. Let
      $\tilde{T}_2 = \hat{T}_1 - \tilde{A} + \tilde{A}'$. Then $\hat{T}_2 = \tilde{T}_2$.
    \end{lemma}
    \begin{proof}
      Suppose that we contract $B$ in $A$ and $A'$ to obtain $\hat{A}$ and
      $\hat{A}'$. It is easy to see that $\hat{A} = \tilde{A}$
      and $\hat{A}' = \tilde{A}'$. So all that remains to show
      is that $\hat{T}_1 - \hat{A} - \hat{A}'$ is equal to $\hat{T}_2$.
      This follows immediately from the fact that $T_2 = T_1 - A + A'$.
    \end{proof}

    Suppose that we initialize the data structure of Lemma~\ref{lem:forest} for
    a tree $T$. Suppose further that we contract the tree down to a single 
    vertex, and when we are contracting an edge $uv$, we call
    $\contract(uv, z)$ where $z$ is the vertex $\argmin_{w \in \{u,v\}} \deg(w)$.
    Lemma~\ref{lem:forest} shows that a single contraction requires time
    roughly proportional to the degree of $z$. The following lemma then shows
    that contracting $T$ down to a single vertex via the data structure of
    Lemma~\ref{lem:forest} requires $\tO(n)$ time. Note that the
    lemma holds for a more general setting.

    \begin{lemma}\label{lem:cont-run}
      Let $G = (V,E)$ be a graph on $m$ edges and $n$ vertices.  
      Let $G_k$ be the graph after $k-1$ edge contractions, so $G_1 = G$.
      Let $u_kv_k$ be the edge contracted in $G_k$ for $k \in [n-1]$.
      We allow any number of edge deletions between edge contractions and
      assume that $G_k$ is the state of the graph just prior to the edge $u_kv_k$
      being contracted.

      Suppose there exists an $\alpha > 0$ such that for all $k \in [n-1]$, the time 
      it takes to contract $u_kv_k$ in $G_k$ is at most
      $\alpha \cdot \min\{\deg_{G_k}(u_k), \deg_{G_k}(v_k)\}$. Then
      the total running time to contract the edges in any order is 
      $O(\alpha m \log n)$.
    \end{lemma}    
    \begin{proof}
      Let $k \in[n-1]$. For $w \in G_k$, $w$ is the result of
      some number of edge contractions and therefore corresponds to a subgraph of
      $G$. Let $V_k(w)$ be the vertices of this subgraph.
      Without loss of generality, assume $\abs{V_k(u_k)} \le \abs{V_k(v_k)}$. Then the
      running time of contraction is at most
      $\alpha\cdot \min\{\deg_{G_k}(u_k), \deg_{G_k}(v_k)\} \le 
       \alpha \cdot \sum_{v \in V_k(u_k)} \deg_G(v)$.
      That is, when contracting $u_kv_k$ in $G_k$, we only pay for the degrees of
      vertices that are in the smaller contracted vertex, where size is with
      respect to the cardinality of $V_k(w)$ for a vertex $w \in G_k$.
      
      So for a vertex $u \in V(G)$, we only pay $\alpha \cdot\deg_G(u)$ time
      whenever contracting the edge $u_kv_k$
      where $u \in V_k(u_k)$ such that $\abs{V_k(u_k)} \le \abs{V_k(v_k)}$.
      Since this can happen at most $O(\log n)$ times, the amount
      of times we pay $\alpha \cdot \deg_G(v)$ for a vertex $v$ is at most $O(\log n)$.
      Therefore, the total running time is
      $O(\alpha\cdot \log n\sum_{v \in V} \deg_G(v)) = O(\alpha m \log n)$.
    \end{proof}

%% file: main.bbl
\newcommand{\etalchar}[1]{$^{#1}$}
\begin{thebibliography}{AGM{\etalchar{+}}17}

\bibitem[AG15]{AnariG15}
N.~{Anari} and S.~O. {Gharan}.
\newblock Effective-resistance-reducing flows, spectrally thin trees, and
  asymmetric tsp.
\newblock In {\em 2015 IEEE 56th Annual Symposium on Foundations of Computer
  Science}, pages 20--39, 2015.

\bibitem[AGM{\etalchar{+}}17]{Asadpouretal17}
Arash Asadpour, Michel~X Goemans, Aleksander M{\'{a}}dry, Shayan~Oveis Gharan,
  and Amin Saberi.
\newblock An {$O(\log n/\log \log n)$}-approximation algorithm for the
  asymmetric traveling salesman problem.
\newblock {\em Operations Research}, 65(4):1043--1061, 2017.

\bibitem[AGV18]{AnariGV18}
Nima Anari, Shayan~Oveis Gharan, and Cynthia Vinzant.
\newblock Log-concave polynomials, entropy, and a deterministic approximation
  algorithm for counting bases of matroids.
\newblock In {\em 2018 IEEE 59th Annual Symposium on Foundations of Computer
  Science (FOCS)}, pages 35--46. IEEE, 2018.

\bibitem[ALGV20]{AnariLOV20}
Nima Anari, Kuikui Liu, Shayan~Oveis Gharan, and Cynthia Vinzant.
\newblock Log-concave polynomials iv: Exchange properties, tight mixing times,
  and faster sampling of spanning trees, 2020.

\bibitem[AS04]{pipage}
Alexander~A. Ageev and Maxim Sviridenko.
\newblock Pipage rounding: A new method of constructing algorithms with proven
  performance guarantee.
\newblock {\em Journal of Combinatorial Optimization}, 8:307--328, 2004.

\bibitem[Ban19]{Bansal19}
Nikhil Bansal.
\newblock On a generalization of iterated and randomized rounding.
\newblock In {\em Proceedings of the 51st Annual ACM SIGACT Symposium on Theory
  of Computing}, pages 1125--1135, 2019.

\bibitem[BGRS04]{BiloGRS04}
Vittorio Bilo, Vineet Goyal, Ramamoorthi Ravi, and Mohit Singh.
\newblock On the crossing spanning tree problem.
\newblock In {\em Approximation, Randomization, and Combinatorial Optimization.
  Algorithms and Techniques}, pages 51--60. Springer, 2004.

\bibitem[BKK{\etalchar{+}}13]{BansalKKNP13}
Nikhil Bansal, Rohit Khandekar, Jochen K{\"o}nemann, Viswanath Nagarajan, and
  Britta Peis.
\newblock On generalizations of network design problems with degree bounds.
\newblock {\em Mathematical Programming}, 141(1-2):479--506, 2013.

\bibitem[BKN10]{BansalKN10}
Nikhil Bansal, Rohit Khandekar, and Viswanath Nagarajan.
\newblock Additive guarantees for degree-bounded directed network design.
\newblock {\em SIAM Journal on Computing}, 39(4):1413--1431, Jan 2010.

\bibitem[CCPV11]{ccpv}
Gruia Calinescu, Chandra Chekuri, Martin Pal, and Jan Vondr{\'a}k.
\newblock Maximizing a monotone submodular function subject to a matroid
  constraint.
\newblock {\em SIAM Journal on Computing}, 40(6):1740--1766, 2011.

\bibitem[CHPQ20]{chq-20}
Chandra Chekuri, Sariel Har-Peled, and Kent Quanrud.
\newblock Fast lp-based approximations for geometric packing and covering
  problems.
\newblock In {\em Proceedings of the Fourteenth Annual ACM-SIAM Symposium on
  Discrete Algorithms}, pages 1019--1038. SIAM, 2020.

\bibitem[CLS{\etalchar{+}}19]{ChakLSSW19}
Deeparnab Chakrabarty, Yin~Tat Lee, Aaron Sidford, Sahil Singla, and Sam
  Chiu-wai Wong.
\newblock Faster matroid intersection.
\newblock In {\em 2019 IEEE 60th Annual Symposium on Foundations of Computer
  Science (FOCS)}, pages 1146--1168. IEEE, 2019.

\bibitem[CNN11]{CharikarNN11}
Moses Charikar, Alantha Newman, and Aleksandar Nikolov.
\newblock Tight hardness results for minimizing discrepancy.
\newblock In {\em Proceedings of the twenty-second annual ACM-SIAM symposium on
  Discrete Algorithms}, pages 1607--1614. SIAM, 2011.

\bibitem[CQ17a]{cq-tsp}
Chandra Chekuri and Kent Quanrud.
\newblock Approximating the held-karp bound for metric tsp in nearly-linear
  time.
\newblock In {\em 2017 IEEE 58th Annual Symposium on Foundations of Computer
  Science (FOCS)}, pages 789--800. IEEE, 2017.

\bibitem[CQ17b]{cq-17}
Chandra Chekuri and Kent Quanrud.
\newblock Near-linear time approximation schemes for some implicit fractional
  packing problems.
\newblock In {\em Proceedings of the Twenty-Eighth Annual ACM-SIAM Symposium on
  Discrete Algorithms}, pages 801--820. SIAM, 2017.

\bibitem[CQ18a]{cq-18b}
Chandra Chekuri and Kent Quanrud.
\newblock Fast approximations for metric-tsp via linear programming.
\newblock {\em arXiv preprint arXiv:1802.01242}, 2018.

\bibitem[CQ18b]{cq-18}
Chandra Chekuri and Kent Quanrud.
\newblock Randomized mwu for positive lps.
\newblock In {\em Proceedings of the Twenty-Ninth Annual ACM-SIAM Symposium on
  Discrete Algorithms}, pages 358--377. SIAM, 2018.

\bibitem[CQT20]{cqt-20}
Chandra Chekuri, Kent Quanrud, and Manuel~R. Torres.
\newblock Fast approximation algorithms for bounded degree and crossing
  spanning tree problems.
\newblock {\em CoRR}, abs/2011.03194, 2020.

\bibitem[CRRT07]{Chaudhurietal07}
Kamalika Chaudhuri, Satish Rao, Samantha Riesenfeld, and Kunal Talwar.
\newblock What would edmonds do? augmenting paths and witnesses for
  degree-bounded msts.
\newblock {\em Algorithmica}, 55(1):157--189, Nov 2007.

\bibitem[CVZ10]{cvz-10}
Chandra Chekuri, Jan Vondrak, and Rico Zenklusen.
\newblock Dependent randomized rounding via exchange properties of
  combinatorial structures.
\newblock In {\em 2010 IEEE 51st Annual Symposium on Foundations of Computer
  Science}, pages 575--584. IEEE, 2010.

\bibitem[CVZ11]{cvz-11}
Chandra Chekuri, Jan Vondr{\'a}k, and Rico Zenklusen.
\newblock Multi-budgeted matchings and matroid intersection via dependent
  rounding.
\newblock In {\em Proceedings of the twenty-second annual ACM-SIAM symposium on
  Discrete Algorithms}, pages 1080--1097. SIAM, 2011.

\bibitem[DHZ17]{DuanHZ17}
Ran Duan, Haoqing He, and Tianyi Zhang.
\newblock Near-linear time algorithms for approximate minimum degree spanning
  trees.
\newblock {\em ArXiv}, abs/1712.09166, 2017.

\bibitem[EN19]{en-19}
Alina Ene and Huy~L Nguyen.
\newblock Towards nearly-linear time algorithms for submodular maximization
  with a matroid constraint.
\newblock In {\em International Colloquium on Automata, Languages, and
  Programming}, volume 132, 2019.

\bibitem[FR94]{FurerR94}
M.~F{\"u}rer and B.~Raghavachari.
\newblock Approximating the minimum-degree steiner tree to within one of
  optimal.
\newblock {\em Journal of Algorithms}, 17(3):409--423, Nov 1994.

\bibitem[GK07]{gk-07}
Naveen Garg and Jochen Koenemann.
\newblock Faster and simpler algorithms for multicommodity flow and other
  fractional packing problems.
\newblock {\em SIAM Journal on Computing}, 37(2):630--652, 2007.

\bibitem[GKPS06]{GandhiKPS06}
Rajiv Gandhi, Samir Khuller, Srinivasan Parthasarathy, and Aravind Srinivasan.
\newblock Dependent rounding and its applications to approximation algorithms.
\newblock {\em Journal of the ACM (JACM)}, 53(3):324--360, 2006.

\bibitem[Goe06]{Goemans06}
Michel Goemans.
\newblock Minimum bounded degree spanning trees.
\newblock {\em 2006 47th Annual IEEE Symposium on Foundations of Computer
  Science (FOCS '06)}, 2006.

\bibitem[GSS11]{GharanSS11}
Shayan~Oveis Gharan, Amin Saberi, and Mohit Singh.
\newblock A randomized rounding approach to the traveling salesman problem.
\newblock In {\em 2011 IEEE 52nd Annual Symposium on Foundations of Computer
  Science}, pages 550--559. IEEE, 2011.

\bibitem[GW17]{gw-17}
Kyle Genova and David~P Williamson.
\newblock An experimental evaluation of the best-of-many christofides'
  algorithm for the traveling salesman problem.
\newblock {\em Algorithmica}, 78(4):1109--1130, 2017.

\bibitem[HDLT01]{hlt-01}
Jacob Holm, Kristian De~Lichtenberg, and Mikkel Thorup.
\newblock Poly-logarithmic deterministic fully-dynamic algorithms for
  connectivity, minimum spanning tree, 2-edge, and biconnectivity.
\newblock {\em Journal of the ACM (JACM)}, 48(4):723--760, 2001.

\bibitem[Kar98]{k-98}
David~R Karger.
\newblock Random sampling and greedy sparsification for matroid optimization
  problems.
\newblock {\em Mathematical Programming}, 82(1-2):41--81, 1998.

\bibitem[KKG20a]{KarlinKG20a}
Anna~R Karlin, Nathan Klein, and Shayan~Oveis Gharan.
\newblock An improved approximation algorithm for tsp in the half integral
  case.
\newblock In {\em Proceedings of the 52nd Annual ACM SIGACT Symposium on Theory
  of Computing}, pages 28--39, 2020.

\bibitem[KKG20b]{KarlinKG20b}
Anna~R Karlin, Nathan Klein, and Shayan~Oveis Gharan.
\newblock A (slightly) improved approximation algorithm for metric tsp.
\newblock {\em arXiv preprint arXiv:2007.01409}, 2020.

\bibitem[KLS12]{KLS12}
Tam{\'a}s Kir{\'a}ly, Lap~Chi Lau, and Mohit Singh.
\newblock Degree bounded matroids and submodular flows.
\newblock {\em Combinatorica}, 32(6):703--720, 2012.

\bibitem[KR03]{kr-03}
Jochen K{\"o}nemann and R~Ravi.
\newblock Primal-dual meets local search: approximating mst's with nonuniform
  degree bounds.
\newblock In {\em Proceedings of the thirty-fifth annual ACM symposium on
  Theory of computing}, pages 389--395, 2003.

\bibitem[KR05]{KR05}
Jochen K{\"o}nemann and Ramamoorthi Ravi.
\newblock Primal-dual meets local search: approximating msts with nonuniform
  degree bounds.
\newblock {\em SIAM Journal on Computing}, 34(3):763--773, 2005.

\bibitem[LS18]{LinharesS18}
Andr{\'e} Linhares and Chaitanya Swamy.
\newblock Approximating min-cost chain-constrained spanning trees: a reduction
  from weighted to unweighted problems.
\newblock {\em Mathematical Programming}, 172(1-2):17--34, 2018.

\bibitem[OZ18]{OlverZ18}
Neil Olver and Rico Zenklusen.
\newblock Chain-constrained spanning trees.
\newblock {\em Mathematical Programming}, 167(2):293--314, 2018.

\bibitem[PS97]{PS97}
Alessandro Panconesi and Aravind Srinivasan.
\newblock Randomized distributed edge coloring via an extension of the
  chernoff-hoeffding bounds.
\newblock {\em SIAM J. Comput.}, 26:350--368, 1997.

\bibitem[Qua18]{q-18}
Kent Quanrud.
\newblock Fast and deterministic approximations for $ k $-cut.
\newblock {\em arXiv preprint arXiv:1807.07143}, 2018.

\bibitem[Qua19]{q-thesis}
Kent Quanrud.
\newblock {\em Fast approximations for combinatorial optimization via
  multiplicative weight updates}.
\newblock PhD thesis, University of Illinois, Urbana-Champaign, 2019.

\bibitem[Sch03]{Schrijver-book}
Alexander Schrijver.
\newblock {\em Combinatorial optimization: polyhedra and efficiency},
  volume~24.
\newblock Springer Science \& Business Media, 2003.

\bibitem[Sch18]{Schild18}
Aaron Schild.
\newblock An almost-linear time algorithm for uniform random spanning tree
  generation.
\newblock In {\em Proceedings of the 50th Annual ACM SIGACT Symposium on Theory
  of Computing}, pages 214--227, 2018.

\bibitem[SL15]{ls-15}
Mohit Singh and Lap~Chi Lau.
\newblock Approximating minimum bounded degree spanning trees to within one of
  optimal.
\newblock {\em Journal of the ACM (JACM)}, 62(1):1--19, 2015.

\bibitem[Sri01]{Srinivasan01}
Aravind Srinivasan.
\newblock Distributions on level-sets with applications to approximation
  algorithms.
\newblock In {\em Proceedings 42nd IEEE Symposium on Foundations of Computer
  Science}, pages 588--597. IEEE, 2001.

\bibitem[STV18]{SvenssonTV18}
Ola Svensson, Jakub Tarnawski, and L{\'a}szl{\'o}~A V{\'e}gh.
\newblock A constant-factor approximation algorithm for the asymmetric
  traveling salesman problem.
\newblock In {\em Proceedings of the 50th Annual ACM SIGACT Symposium on Theory
  of Computing}, pages 204--213, 2018.

\bibitem[SV14]{SinghV14}
Mohit Singh and Nisheeth~K Vishnoi.
\newblock Entropy, optimization and counting.
\newblock In {\em Proceedings of the forty-sixth annual ACM symposium on Theory
  of computing}, pages 50--59, 2014.

\bibitem[Tar75]{t-75}
Robert~Endre Tarjan.
\newblock Efficiency of a good but not linear set union algorithm.
\newblock {\em Journal of the ACM (JACM)}, 22(2):215--225, 1975.

\bibitem[Tar83]{t-83}
Robert~Endre Tarjan.
\newblock {\em Data structures and network algorithms}, volume~44.
\newblock Siam, 1983.

\bibitem[TV20]{TraubV20}
Vera Traub and Jens Vygen.
\newblock An improved approximation algorithm for atsp.
\newblock In {\em Proceedings of the 52nd Annual ACM SIGACT Symposium on Theory
  of Computing}, pages 1--13, 2020.

\bibitem[Wan17]{w-thesis}
Di~Wang.
\newblock {\em Fast Approximation Algorithms for Positive Linear Programs}.
\newblock PhD thesis, EECS Department, University of California, Berkeley, Jul
  2017.

\bibitem[Win89]{Win89}
Sein Win.
\newblock On a connection between the existence ofk-trees and the toughness of
  a graph.
\newblock {\em Graphs and Combinatorics}, 5(1):201--205, 1989.

\end{thebibliography}
